\documentclass[twocolumn]{IEEEtran}

\ifCLASSINFOpdf
\else
\fi
\usepackage{amsmath,amsfonts}

\usepackage[linesnumbered,ruled,vlined]{algorithm2e}
\usepackage{array}
\usepackage{textcomp}
\usepackage{stfloats}
\usepackage{url}
\usepackage{verbatim}
\usepackage{graphicx}
\usepackage{cite}

\usepackage{subcaption} 

\usepackage{amsmath,amssymb,amsfonts}

\usepackage{bm}

\usepackage{textcomp}
\usepackage[table, dvipsnames]{xcolor} 

\usepackage{pbox}

\usepackage{amsthm}

\theoremstyle{plain}
\newtheorem{lemma}{Lemma}
\newtheorem{theorem}{Theorem}
\newtheorem{remark}{Remark}

\DeclareMathOperator*{\argmax}{arg\,max}

\begin{document}


\title{Unifying Performance and Equity via Percentile Optimization and Cell-Free Massive MIMO}

\author{Ahmad Khan\IEEEauthorrefmark{1},~\IEEEmembership{Member,~IEEE}, Hussein~A.~Ammar\IEEEauthorrefmark{2},~\IEEEmembership{Member,~IEEE}, {Raviraj Adve}\IEEEauthorrefmark{3},~\IEEEmembership{Fellow,~IEEE}
\thanks{
    \IEEEauthorrefmark{1}Ahmad Khan is with Ericsson Ottawa R\&D, 349 Terry Fox Drive, Kanata, K2K 2V6. 
}
\thanks{
    \IEEEauthorrefmark{2}Hussein A. Ammar is with the Department of Electrical and Computer Engineering at the Royal Military College of Canada, 13 General Crerar Crescent, Kingston, ON K7K 7B4.
}
\thanks{
    \IEEEauthorrefmark{3}{Raviraj Adve is with the Department of Electrical and Computer Engineering at the University of Toronto, 10 Kings College Road, M5S 3G4.\newline Both Hussein A. Ammar and Ahmad Khan contributed equally to this work.}
}
\vspace{-1em}
}

\maketitle

\begin{abstract}
Improving the throughput of the weakest users is a central objective of wireless resource allocation, and the objective around which the user-centric cell-free massive multiple-input multiple-output (UC-mMIMO) architecture is built. This paper brings the two together by formulating and solving percentile programs for UC-mMIMO networks. We maximize the sum-least-$q^{\rm th}$-percentile (SLqP) rate via downlink power control, a utility that targets the user equipments (UEs) in the lowest $q$ percentile of the rate distribution and subsumes max-min-rate and sum-rate maximization as special cases. The problem is non-convex, non-smooth and NP-hard. The cyclic maximization approach of one of the previous percentile optimization series does not carry over: the coherent joint transmission technique introduces cross terms in the powers that destroy the block-concavity on which it relies. We develop a change of variables that restores block-concavity, yielding an algorithm with closed-form auxiliary updates and a concave power update, and we give a self-contained proof of non-decreasing convergence to a directional stationary point. Simulations establish two facts: first, percentile optimization retains its value in an architecture with no cell edges; relative to the weighted minimum mean-square error benchmark, it raises the average rate of the lowest-$10^{\rm th}$-percentile UEs by 56--72\% and Jain's fairness index from below 0.92 to above 0.99. Second, compared directly with massive MIMO and small cells under the same SLqP utility, UC-mMIMO delivers the highest lower-percentile rates at every access point density tested, by up to 62\% and 2.7$\times$ respectively, at the same fairness, and its advantage widens with densification.
\end{abstract}


\section{Introduction}
User fairness reflects how evenly the network resources are distributed among user equipment (UEs), and it is a key objective for operators, since it ensures that no UE is left with a service quality substantially below the target quality of service (QoS) while others are allocated more than they require~\cite{7917284}. Fairness is of particular concern for the cell-edge UEs of conventional networks, which occupy the lowest $5\%$ or $10\%$ of the rate distribution~\cite{9084256}, and treating these UEs equitably requires that the network resources be allocated with fairness as an explicit design objective.

The user-centric cell-free massive multiple-input multiple-output (UC-mMIMO) architecture has been proposed for beyond-5G networks with the aim of eliminating the cell boundaries of conventional cellular networks: each user equipment (UE) is served by a set of neighbouring, cooperating access points (APs)~\cite{interdonato2019ubiquitous, 7827017, 9064545, surveyCF_9650567}, so that each UE is at the center of its own cluster and has its own ``virtual~cell''~\cite{surveyCF_9650567}. This marks a clear departure from the cellular structure that has governed wireless network design for decades. It improves, in particular, the performance of ``cell-edge'' UEs, which consistently experience the poorest coverage and the lowest data rates, i.e., the lowest rate percentiles. A gain of~20~dB in the $5^\mathrm{th}$-percentile signal-to-interference-plus-noise ratio (SINR) can be achieved relative to the co-located massive MIMO scheme adopted in 4G-advanced and 5G networks \cite{demir2021foundations}.

\begin{figure}[t]
	\centering
    \includegraphics[width=1\linewidth]{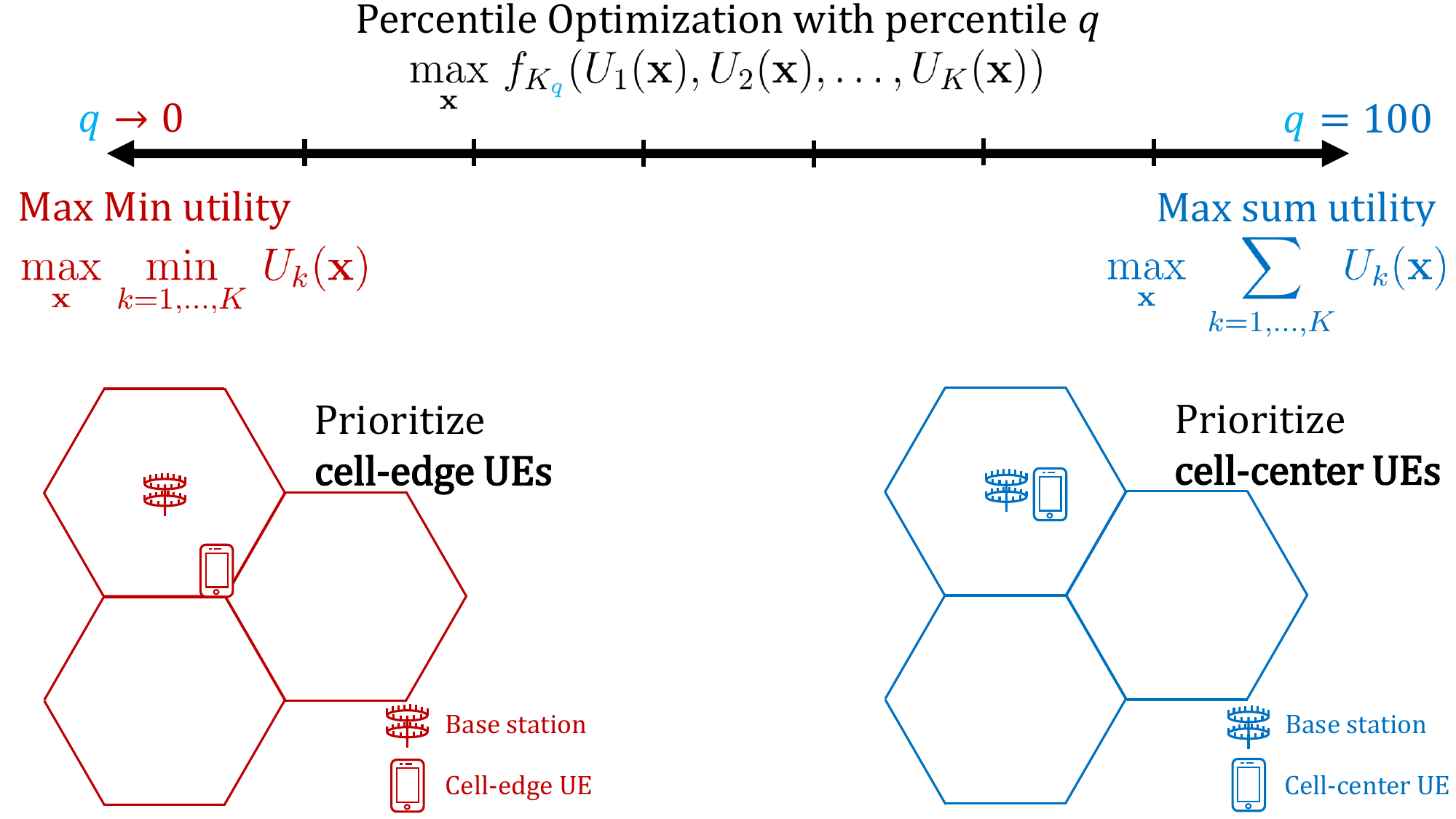}
	\caption{Percentile programs family. The percentile $q$ interpolates between max-min-rate maximization ($q \to 0$), which favors ``conventional cell-edge'' UEs, and sum-rate maximization ($q = 100$), which favors ``conventional cell-center UEs''.}
	\label{fig:percentile_qnty}
\end{figure}

\vspace{-1em}
\subsection{Related Work and Research Gap}

Resource allocation in UC-mMIMO networks has been studied under a variety of utilities, most commonly the minimum rate, the weighted sum-rate and the sum-rate. The max-min utility~\cite{7827017} guarantees a reasonable level of service to the worst-performing UE. A weighted-sum utility~\cite{9519163, 9570126} allows the network to prioritize particular UEs, with weights typically updated in each time slot according to a fairness rule such as proportional fairness~\cite{10163872}. The sum-rate utility~\cite{9489284} yields the highest aggregate throughput, but is unfair to weaker UEs, which are assigned little or no power and are consequently driven into outage~\cite{PercI_10466635}. 

None of these utilities explicitly targets the lower-percentile UEs. Percentile programs are a recently introduced class of optimization problems that do~\cite{PercI_10466635, PercII_10458868}. Rather than the single worst UE or the aggregate, percentile optimization targets the set of UEs occupying the lowest $q$ percentile of the utility distribution. Crucially, the class subsumes the utilities above: max-min-rate and sum-rate maximization are recovered as the two extreme members of the family, and the percentile $0 < q \le 100$ spans between them, as illustrated in Fig.~\ref{fig:percentile_qnty}, giving the network operator direct control over the inherent tradeoff. Optimizing for the lowest $q$ percentile is, moreover, distinct from optimizing for the single worst UE, and the two lead to different outcomes~\cite{PercI_10466635}. Catering to the weakest UE alone incurs a large rate penalty for most other UEs and is fragile: a UE with a very poor channel consumes a disproportionate share of the resources~\cite{5595728} and, as the number of UEs grows without bound, the optimal max-min objective converges to zero almost surely~\cite{8972604}. Percentile optimization instead improves a broader set of low-performing UEs and strikes a more favourable balance between fairness and aggregate utility~\cite{PercII_10458868}.

Alternative utilities have also been proposed with the aim of prioritizing cell-edge UEs in conventional networks, notably $\alpha$-fair power allocation~\cite{7111370}, in which a scalar $\alpha$ adjusts the degree of fairness. As observed in~\cite{5461911}, however, such measures are not physically interpretable: they specify neither the degree of fairness being sought nor the percentile of UEs being targeted, and it is unclear in what sense a problem with $\alpha = 3$ is fairer than one with $\alpha = 2$. From the perspective of percentile optimization, $\alpha$-fairness is a heuristic, as is weighted-sum utility with proportional-fair weights.
The two-part study in~\cite{PercI_10466635, PercII_10458868} removes this ambiguity by optimizing lower-percentile throughput explicitly, so that the degree of fairness becomes a design input rather than a by-product of a tuning parameter: Part~I~\cite{PercI_10466635} addresses power control for a multicell, multiuser single-input single-output network, and Part~II~\cite{PercII_10458868} beamforming for the multiuser MIMO setting.

All of the studies above were developed for conventional cellular networks, in which each UE is served by a single base station (BS). Such an architecture does not prioritize cell-edge UEs; indeed, the cellular structure is itself one of the principal obstacles to improving cell-edge performance, since it gives rise to regions in which the desired and interfering signal powers are comparable~\cite{10522673}. UC-mMIMO, in contrast, is designed around precisely these UEs~\cite{demir2021foundations}. It is therefore natural to ask what is obtained when percentile optimization is combined with an architecture that already favours the weaker UEs, and in particular \textit{whether the resource allocation still matters once the architecture has removed the cell edges}. To the best of our knowledge, this question has not previously been studied, and it is the gap that the present work addresses.

\subsection{Contributions}

In this work, we formulate and solve percentile programs for the UC-mMIMO architecture: we optimize the downlink transmit powers for a sum-least-$q^{\rm th}$-percentile (SLqP) utility that targets the lowest $q$ percentile of UEs explicitly, a problem that is non-convex, non-smooth and NP-hard. Our contributions are threefold:
\begin{itemize}
\item \textit{New algorithm for the cell-free setting, with a complete convergence analysis:} We show that the cyclic maximization approach of \cite{PercI_10466635} cannot be applied directly to UC-mMIMO systems. Coherent joint transmission from the APs of a cluster introduces cross products of the powers of distinct APs into every rate expression, and these terms destroy the block-concavity of the fractionally transformed problem on which the approach of \cite{PercI_10466635} rests. 

We develop a change of variables that restores block-concavity exactly, and combine it with the quadratic fractional transform and an epigraph reformulation to obtain a cyclic maximization algorithm with closed-form auxiliary updates and a concave power update; the SLqP utility is optimized directly, without smoothing, relaxation or surrogate approximation. 

We then give a self-contained convergence analysis that accounts for the change of variables: the algorithm is non-decreasing in the SLqP rate, and every limit point is a directional stationary point, which we prove by establishing that the minorant constructed in the transformed domain is tangent to the non-smooth SLqP objective, including in its directional derivatives. The additional step is exactly what separates multi-node from single-node architectures: for the massive MIMO and small-cell benchmarks, our algorithm reduces to the scheme of \cite{PercI_10466635}.

\item \textit{Percentile optimization retains its value in a network with no cell edges:} One might expect an architecture built to eliminate cell-edge UEs to leave little for a cell-edge-oriented resource allocation to do. Our results show otherwise. Relative to the weighted minimum mean-square error (WMMSE) benchmark~\cite{5756489}, the proposed algorithm improves the average rate of the lowest-$10^{\rm th}$-percentile UEs in a UC-mMIMO network by 56--72\% across AP densities from 75 to 250~APs/km$^2$, and raises Jain's fairness index from 0.80--0.92 to above 0.99. Strikingly, a conventional massive MIMO network under percentile optimization delivers about as much to its lowest-percentile UEs as a UC-mMIMO network under WMMSE, and at moderate densities considerably more: the resource allocation is worth as much as the architecture, and the two gains compound.

\item \textit{Cell-free is the superior architecture for lower-percentile service:} We compare UC-mMIMO, co-located massive MIMO and small cells under the same SLqP utility, with the same AP and UE locations and the same total power, so that the comparison isolates the effect of the architecture. UC-mMIMO delivers the highest lower-percentile rates at every AP density and every cell size considered. Its advantage over massive MIMO grows from 8\% at 75~APs/km$^2$ to 62\% at 250~APs/km$^2$, and over small cells from 1.7$\times$ to 2.7$\times$, while all three architectures attain the same fairness; UC-mMIMO is thus the architecture that most effectively converts densification into the best lower-percentile throughput.
\end{itemize}
In addition to the above, we quantify the computational complexity and the fronthaul signalling load of the proposed scheme, and show that the latter is incurred once per coherence block and does not grow with the number of iterations.

\textit{Organization and Notation:} Section~\ref{sec:system_model} presents the system model, Section~\ref{sec:problem_formulation} formulates the problem, Section~\ref{sec:proposed_approach} develops the proposed algorithm and its analysis, Section~\ref{sec:results} presents numerical results and Section~\ref{sec:conclusion} concludes the paper. 

Boldface lower and upper case letters, e.g., ${\bf a}$ and ${\bf A}$, denote column vectors and matrices; ${\bf I}_M$ is the identity matrix of size $M$ and ${\bf 1}_K$ the all-ones vector of length $K$; $(\cdot)^T$, $(\cdot)^*$ and $(\cdot)^H$ denote the transpose, conjugate and conjugate transpose; $|\mathcal{S}|$ is the cardinality of the set $\mathcal{S}$; $\|\cdot\|_2$ and $|\cdot|$ denote the Euclidean norm and the modulus; ${\bm x} \sim \mathcal{CN}({\bm \mu}, {\bf R})$ denotes a circularly symmetric complex Gaussian vector with mean ${\bm \mu}$ and covariance ${\bf R}$; and ${\bf x}^{\uparrow}$ denotes the vector ${\bf x}$ sorted in ascending order, with $x_i^{\uparrow}$ its $i^{\rm th}$ entry.

\begin{figure}[t]
	    \centering
	\includegraphics[width=\linewidth, trim=0 0 329 0, clip]{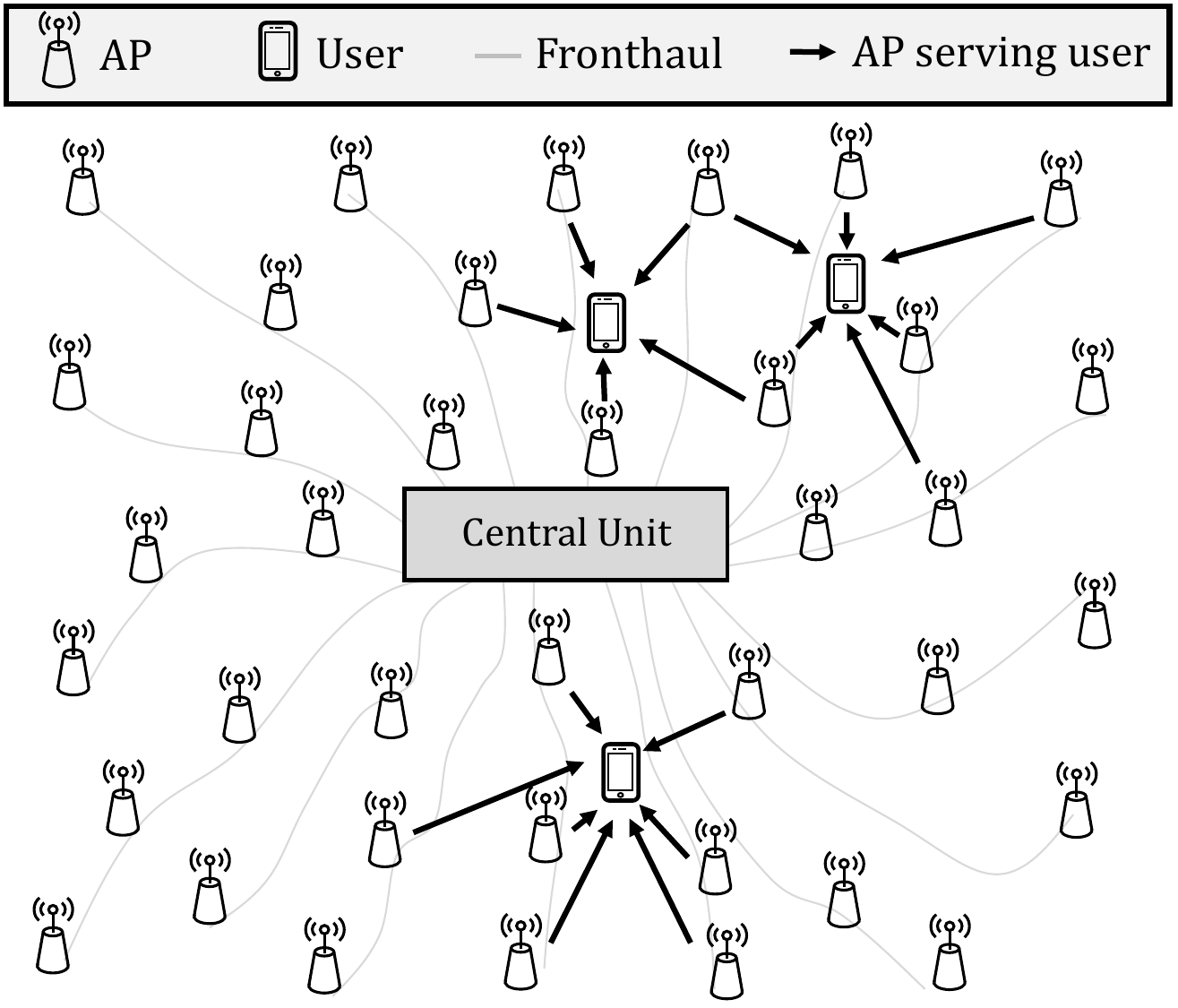}
	\caption{A user-centric cell-free massive MIMO network.}
	\label{fig:UC_singleC}
\end{figure}

\section{System Model}\label{sec:system_model}

We consider the downlink of a UC-mMIMO network comprising a set of APs $\mathcal{B} = \{1, \dots, B\}$ and a set of active UEs $\mathcal{K} = \{1, \dots, K\}$, indexed network-wide by $b$ and $k$, respectively. Every AP is equipped with $M$ antennas and every UE with a single antenna. All APs are connected through a wired fronthaul to a central unit (CU), which performs the resource allocation. The network is shown in Fig.~\ref{fig:UC_singleC}.

We assume operation in the sub-6~GHz band and a block-fading channel model, in which the channel between AP $b$ and UE $k$ is ${\bf h}_{k,b} = \sqrt{\beta_{k,b}}\, {\bf g}_{k,b}$, where $\beta_{k,b}$ is the large-scale fading (LSF) coefficient, assumed constant across many coherence intervals, and ${\bf g}_{k,b}$ is the small-scale fading. Each UE $k$ is served by a cluster of neighbouring APs $\mathcal{C}_k \subseteq \mathcal{B}$, and each AP $b$ serves the set of UEs $\mathcal{K}_b \subseteq \mathcal{K}$; the two are dual descriptions of the same association, in the sense that $b \in \mathcal{C}_k \Leftrightarrow k \in \mathcal{K}_b$. The clusters are formed irrespective of any cellular structure. We populate the cluster of UE $k$ with every AP whose LSF toward UE $k$ exceeds a threshold $\rho$, together with the strongest AP so that no UE is left unserved when the threshold is met by no AP, i.e.,
\begin{align}\label{eq:C_k}
\mathcal{C}_{k}
&=
\{b \in \mathcal{B}: \beta_{k,b} \ge \rho \} \cup \Big\{ \argmax_{b \in \mathcal{B}}\ \beta_{k,b} \Big\}
\end{align}
where the threshold $\rho$ is obtained from the relevant path loss model ${\rm PL}(\cdot)$ and a distance threshold $\bar{d}$ as $\rho = {\rm PL}(\bar{d})$. To focus on the architecture and the combination with percentile optimization, we assume perfect channel state information is available to all APs. This assumption can be easily relaxed~\cite{9519163}.

AP $b$ transmits the signal ${\bf x}_b = \sum_{k \in \mathcal{K}_b} \sqrt{p_{k,b}}\, {\bf v}_{k,b}\, s_{k}$, where $s_k$ is the unit-power data symbol intended for UE $k$, $p_{k,b}$ is the power that AP $b$ allocates to UE $k$, and ${\bf v}_{k,b}$ is the unit-norm beamformer of AP $b$ for UE $k$, which we fix throughout to the maximum-ratio transmission (MRT) direction
\begin{equation}\label{eq:MRT}
\mathbf{v}_{k,b}=\frac{\mathbf{h}_{k,b}}{\left\Vert \mathbf{h}_{k,b}\right\Vert _{2}}
\end{equation}
Each AP is subject to the power budget $\sum_{k\in\mathcal{K}_{b}}p_{k,b}\leq P_{\mathrm{max}}$, and the powers $\{p_{k,b}\}$ are the only design variables. UE $k$ receives \[y_k = \sum_{b \in \mathcal{B}} {\bf h}_{k,b}^H {\bf x}_b + z_k = a_{k,k}(\mathbf{P}) s_k + \sum_{k' \in \mathcal{K}_{-k}} a_{k,k'}(\mathbf{P}) s_{k'} + z_k\]
comprising the desired signal, the interference and the noise, where we overload our notation to define $\mathcal{K}_{-k} \doteq \mathcal{K} \backslash \{k\}$ as the set of all UEs other than $k$, where all transmit powers are collected in $\mathbf{P}$, and where \[a_{k,k'}(\mathbf{P}) = \sum_{b \in\mathcal{C}_{k'}} \sqrt{p_{k',b}}\, {\bf h}_{k,b}^H {\bf v}_{k',b}\] is the precoded channel observed by UE $k$ when the cluster $\mathcal{C}_{k'}$ transmits to UE $k'$, and $z_k \sim \mathcal{CN}(0, \sigma^2)$ is the receiver noise. The summation over the serving cluster is absorbed into $a_{k,k'}(\mathbf{P})$, which reflects the coherent combination of the transmissions of the APs in $\mathcal{C}_{k'}$. The achievable rate of UE $k$ is therefore
\begin{equation}\label{rate_function_cell_free}
r_{k}\left(\mathbf{P}\right)=\mathrm{log}\left(1+
\frac{\left| a_{k,k}(\mathbf{P})\right|^{2}}
{ \underset{k' \in \mathcal{K}_{-k}}{\sum}  \left| a_{k,k'}(\mathbf{P})\right|^{2}
+\sigma^{2}}
\right)
\end{equation}
We note that the coherent sums over the clusters appear in both the signal term $|a_{k,k}(\mathbf{P})|^{2}$ and the interference terms $|a_{k,k'}(\mathbf{P})|^{2}$ of~\eqref{rate_function_cell_free}, each being the squared modulus of a sum of square roots of the powers; as we show in Section~\ref{sec:proposed_approach}, this is the feature that distinguishes the cell-free problem from its cellular counterpart.

\section{Problem Formulation}\label{sec:problem_formulation}

\subsection{Preliminaries}

A percentile $q \in (0, 100]$ indicates the relative standing of a value within a set: a value lying at the $q = 80$ percentile is at least as large as $80\%$ of the values in that set. Given an integer $K \in \mathbb{Z}_{++}$ and the percentile $q$, we define the \emph{percentile number} as $K_q = \min \{k \in \mathbb{Z}_{++} \mid 100 k/K \ge q \} \le K$, i.e., the smallest number of values for which the selected fraction attains $q$; for example, $K = 6$ and $q = 100/3$ gives $K_q = 2$.

Following ~\cite{PercI_10466635}, we define the SLqP utility function $f_{K_q}({\bf x}): \mathbb{R}^{K} \mapsto \mathbb{R}$ as the sum of the smallest $K_q$ components of its argument, i.e., $f_{K_q}({\bf x}) = \sum_{i = 1}^{K_q} x_i^{\uparrow}$. The properties of $f_{K_q}$ that we rely upon in the sequel follow directly from an equivalent variational form. For any integer $m \le K$, define the set of binary selection vectors of length $K$ with Hamming weight $m$ as $\mathcal{A}_{m} := \{{\bf a} \in \{0, 1\}^{K} \mid {\bf 1}_K^T {\bf a} = m\}$. Each ${\bf a} \in \mathcal{A}_{m}$ selects a subset of $m$ entries of ${\bf x}$, and ${\bf a}^T {\bf x}$ is the sum of the selected entries; the set has cardinality $\binom{K}{m}$. Summing the $K_q$ smallest entries of ${\bf x}$ is equivalent to selecting the subset of size $K_q$ whose sum is smallest, so that
\begin{align}\label{eq:SLqP_2}
f_{K_q}({\bf x}) = \min_{{\bf a} \in \mathcal{A}_{K_q}} {\bf a}^T {\bf x}= {\bf 1}_K^T {\bf x}
-\max_{{\bf a} \in \mathcal{A}_{K-K_q}} {\bf a}^T {\bf x}
\end{align}
We emphasize that~\eqref{eq:SLqP_2} is a characterization rather than the definitive computational recipe: the number of selection vectors grows combinatorially in $K$, and a central objective of Section~\ref{sec:proposed_approach} is to optimize $f_{K_q}$ without enumerating them.

The properties of $f_{K_q}$ on which the sequel relies were established in~\cite{PercI_10466635}; we collect them in the following lemma and refer the reader to~\cite{PercI_10466635} for the proofs.

\begin{lemma}[Properties of the SLqP function~\cite{PercI_10466635}]\label{lem:slqp}
For all $q \in (0, 100]$, the SLqP utility function $f_{K_q}$ is (i)~concave, (ii)~continuous on $\mathbb{R}^{K}$ and (iii)~non-decreasing in each of its arguments. For $K_q \ne K$ it is, however, (iv)~non-smooth: away from the \emph{border points} $\mathcal{X}_{\rm b} := \{{\bf x} \mid x_{K_q}^{\uparrow} = x_{K_q+1}^{\uparrow}\}$, at which the $K_q^{\rm th}$ and $(K_q+1)^{\rm th}$ smallest entries coincide, the partial derivative $\partial f_{K_q}({\bf x})/\partial x_i$ equals $1$ if $x_i \le x_{K_q}^{\uparrow}$ and $0$ otherwise, so that the gradient is piecewise constant, jumps across $\mathcal{X}_{\rm b}$ and does not exist on $\mathcal{X}_{\rm b}$.
\end{lemma}

Properties (i)--(iii) are immediate from~\eqref{eq:SLqP_2}, which exhibits $f_{K_q}$ as the pointwise minimum of finitely many linear functions with non-negative coefficients~\cite{boyd2004convex}; property (iv) reflects the fact that on each region of $\mathbb{R}^{K}$ in which the ordering of the entries of ${\bf x}$ is fixed, $f_{K_q}$ coincides with a single linear function, and that an entry which leaves the set of the $K_q$ smallest ceases to contribute to the sum.

Property (iv) rules out the direct application of first-order methods such as steepest ascent and of second-order methods such as Newton's method, since neither the gradient nor the Hessian exists at the border points. Overcoming this obstacle without resorting to smoothing is a central concern of Section~\ref{sec:proposed_approach}. Fig.~\ref{fig:nonSmooth} plots typical values of $f_{K_q}({\bf x})$ against two of its arguments while the remaining arguments are held fixed: the surface is continuous and concave, and the visible creases are precisely the border points at which the set of the $K_q$ smallest entries changes.

\begin{figure}[t]
	\centering
	\includegraphics[width=\linewidth, trim=5 0 15.5 16, clip]{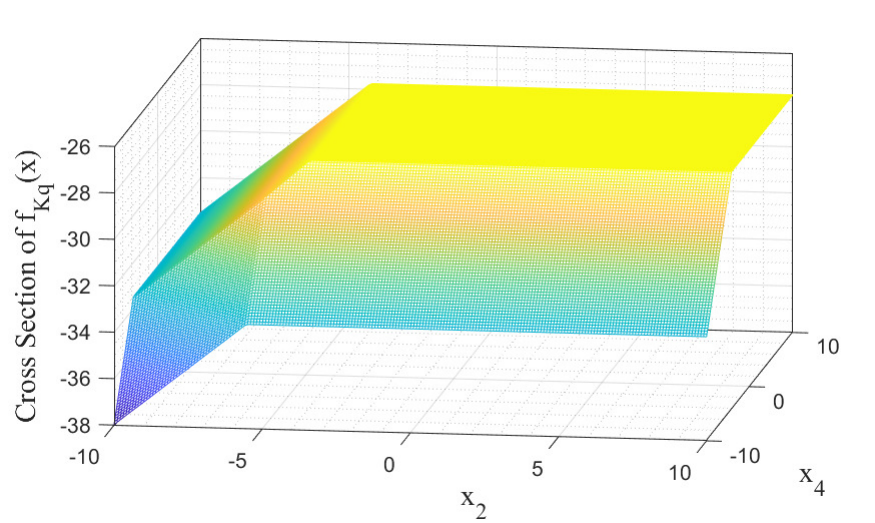}
	\caption{The SLqP function versus two of its arguments, $x_2$ and $x_4$, with the remaining arguments fixed at randomly chosen values; $K_q = 4$ and ${\bf x} \in \mathbb{R}^{10 \times 1}$. The creases are the border points at which the function is non-differentiable.}
	\label{fig:nonSmooth}
\end{figure}

Together with concavity, the monotonicity in Lemma~\ref{lem:slqp}(iii) is what permits the composition arguments, the surrogate bounds and the epigraph reformulation developed in Section~\ref{sec:proposed_approach}, each of which requires the outer aggregating function to be concave and non-decreasing. We also note that $f_{K_q}$ contains the sum and minimum utilities as special cases: $q = 100$ gives $K_q = K$ and $f_{K}({\bf x}) = {\bf 1}_{K}^{T}{\bf x}$, whereas $q = 100/K$ gives $K_q = 1$ and $f_{1}({\bf x}) = \min_{i} x_i$, so that the problem formulated below subsumes sum-rate and max-min-rate maximization as the two extreme members of a single family.

\subsection{Problem Definition and Analysis}

We are now ready to formulate the SLqP rate maximization problem for the UC-mMIMO architecture. Under the system model of Section~\ref{sec:system_model}, our aim can be stated concisely: working in conjunction with the MRT beamforming in~\eqref{eq:MRT}, how should the transmit powers be allocated so as to maximize the network throughput at a desired percentile $q$? The problem that encapsulates this goal is
\begin{align}
\underset{\mathbf{P}}{\mathrm{maximize}}\quad
& f_{K_{q}}\left(r_{1}\left(\mathbf{P}\right),\ldots,r_{K}\left(\mathbf{P}\right)\right)
\tag{P1}\label{SLqP_original_problem}
\\
\mathrm{subject\ to}\quad
& p_{k,b}\geq0,\quad k=1,\ldots,K,\; b=1,\ldots,B
\nonumber
\\
& \sum_{k\in\mathcal{K}_{b}}p_{k,b}\leq P_{\mathrm{max}},\quad b=1,\ldots,B
\nonumber
\end{align}
where the constraints enforce the non-negativity of the powers and the per-AP budget. As in the cellular setting, the choice of $q$ controls the tradeoff between cell-center and cell-edge UEs; with $q=10$, for instance, Problem~\ref{SLqP_original_problem} optimizes the $10^{\rm th}$-percentile throughput, which is among the measures of cell-edge rate most widely adopted in the literature.

Problem~\ref{SLqP_original_problem} inherits the two difficulties of the percentile programs of \cite{PercI_10466635}. First, the objective is non-concave: from~\eqref{eq:SLqP_2}, it is the pointwise minimum, over all $\binom{K}{K_q}$ subsets of size $K_q$, of partial sum-rates, each a sum of logarithms of SINRs in which the powers appear in both numerator and denominator, and a pointwise minimum preserves concavity but cannot create it. Second, it is non-smooth by Lemma~\ref{lem:slqp}(iv) wherever the identity of the $K_q$ worst UEs changes. The two difficulties interlock, since each rules out the natural remedy for the other: a concave surrogate around the current iterate requires the gradient information that the non-smoothness denies, while smoothing the problem by introducing a slack variable together with one constraint per candidate active set, as in the max-min and fractional programming (FP) literature, introduces $\binom{K}{K_q}$ non-convex constraints and is intractable even for small networks~\cite{PercI_10466635}.

A third difficulty is specific to UC-mMIMO. The coherent combination across the APs of a cluster, i.e., the sum over $\mathcal{C}_{k'}$ inside each $a_{k,k'}(\mathbf{P})$ in~\eqref{rate_function_cell_free}, couples the powers of distinct APs serving the same UE through cross products $\sqrt{p_{k,b_1} p_{k,b_2}}$. These products are absent from the rate expressions of a conventional network, in which each UE is served by a single node, and we show in Section~\ref{sec:reformulation} that they defeat the direct application of the approach of \cite{PercI_10466635}.

Finally, the SLqP utility is a \textit{global} objective: identifying the lowest $K_q$ rates requires comparing every UE against every other, so an AP acting on local channel knowledge alone cannot determine whether the UEs it serves belong to the percentile set, and no purely local rule can optimize network percentile utility. We therefore assume throughout that the CU aggregates the requisite channel statistics and coordinates the allocation centrally, and we quantify the fronthaul cost of this assumption in Section~\ref{sec:complexity}.

\section{Proposed Approach}\label{sec:proposed_approach}

We convert Problem~\ref{SLqP_original_problem} into an equivalent block-concave form, thereby enabling cyclic maximization to a stationary point, in three steps. First, the quadratic fractional transform (QFT) decouples the numerator from the denominator of each SINR at the cost of one auxiliary variable per UE. Second, a change of variables on the powers restores the concavity of the transformed rates in the power block; this step is specific to the coherent joint transmission of UC-mMIMO and, as we show in Section~\ref{sec:benchmark_architectures}, is not required in the single-node setting of \cite{PercI_10466635}. Third, an epigraph formulation exposes the SLqP structure to a disciplined convex solver while side-stepping the combinatorial number of non-convex constraints discussed above. Throughout, the SLqP utility is optimized directly; at no stage is it smoothed, relaxed or replaced by a proxy.

\subsection{The Quadratic Fractional Transform}\label{sec:QFT}

Our approach utilizes the QFT, originally proposed by Benson~\cite{benson2004global} and subsequently exploited by Shen and Yu~\cite{8314727} for weighted sum-rate maximization.

\begin{lemma}\label{lemma:QFT}
Let $A(\mathbf{p}): \mathbb{R}^{N} \mapsto \mathbb{R}_{+}$, $B(\mathbf{p}): \mathbb{R}^{N} \mapsto \mathbb{R}_{++}$ and $\mathcal{P} \subseteq \mathbb{R}^{N}$. Then the fractional problem $\max_{\mathbf{p} \in \mathcal{P}} A(\mathbf{p})/B(\mathbf{p})$ is equivalent, in the sense of both the optimal variables and the optimal objective value, to the problem of maximizing the auxiliary objective
\begin{align}\label{eq:QFT_auxiliary_objective}
2x\sqrt{A\left(\mathbf{p}\right)}-x^{2}B\left(\mathbf{p}\right)
\end{align}
over $\mathbf{p} \in \mathcal{P}$ and the auxiliary variable $x \in \mathbb{R}$.
\end{lemma}

\begin{proof}
The auxiliary objective in~\eqref{eq:QFT_auxiliary_objective} is a concave quadratic in $x$, maximized at $x^{\mathrm{opt}}=\sqrt{A(\mathbf{p})}/B(\mathbf{p})$, and substituting this value back into~\eqref{eq:QFT_auxiliary_objective} recovers the original objective, since $2\sqrt{A(\mathbf{p})}\sqrt{A(\mathbf{p})}/B(\mathbf{p}) - A(\mathbf{p})B(\mathbf{p})/B(\mathbf{p})^{2} = A(\mathbf{p})/B(\mathbf{p})$.
\end{proof}

Two consequences of Lemma~\ref{lemma:QFT} are used repeatedly. Since $x^{\mathrm{opt}}$ maximizes the auxiliary objective, for any $\mathbf{p}$ the auxiliary objective lower bounds the ratio,
\begin{align}\label{eq:QFT_lower_bound}
2x\sqrt{A\left(\mathbf{p}\right)}-x^{2}B\left(\mathbf{p}\right) \le \frac{A\left(\mathbf{p}\right)}{B\left(\mathbf{p}\right)},\quad \forall x
\end{align}
and the bound holds with equality at $x = x^{\mathrm{opt}}$.

Two features of Problem~\ref{SLqP_original_problem} prevent the direct application of Lemma~\ref{lemma:QFT}: the ratios to be transformed appear inside a logarithm, and the resulting logarithms are aggregated by the non-smooth SLqP utility rather than by a plain sum. The following result establishes that neither feature obstructs the transform, and that the reason is precisely the monotonicity in Lemma~\ref{lem:slqp}(iii).

\begin{theorem}\label{thm:equivalence}
Let $\phi: \mathbb{R}^{K} \mapsto \mathbb{R}$ be non-decreasing in each of its arguments, and let $\hat{r}_{k}$ denote the auxiliary rate obtained from $r_{k}$ in (\ref{rate_function_cell_free}) by replacing its SINR with the auxiliary objective in~\eqref{eq:QFT_auxiliary_objective}. Then
\begin{align}\label{eq:aggregator_bound}
\phi\big(\hat{\mathbf{r}}(\mathbf{X},\mathbf{P})\big) \le \phi\big(\mathbf{r}(\mathbf{P})\big),\quad \forall \mathbf{X}
\end{align}
with equality when $\mathbf{X} = \mathbf{X}^{\mathrm{opt}}(\mathbf{P})$, obtained by applying the maximizer $x^{\mathrm{opt}}$ of Lemma~\ref{lemma:QFT} to each UE. Consequently, $\max_{\mathbf{X},\mathbf{P}}\, \phi(\hat{\mathbf{r}}(\mathbf{X},\mathbf{P})) = \max_{\mathbf{P}}\, \phi(\mathbf{r}(\mathbf{P}))$, and the two problems are equivalent in the sense of the optimal powers and objective values.
\end{theorem}

\begin{proof}
From~\eqref{eq:QFT_lower_bound}, the auxiliary objective never exceeds the SINR it replaces. Since $\log(1+\cdot)$ is non-decreasing, $\hat{r}_{k} \le r_{k}$ componentwise for every $\mathbf{X}$, and since $\phi$ is non-decreasing in each argument, the inequality is preserved by $\phi$, which establishes~\eqref{eq:aggregator_bound}. Equality at $\mathbf{X}^{\mathrm{opt}}$ follows from the tightness of the bound at $x^{\mathrm{opt}}$. Since the upper bound is attained for every $\mathbf{P}$, the two maximizations coincide.
\end{proof}

By Lemma~\ref{lem:slqp}(iii), the SLqP utility satisfies the hypothesis of Theorem~\ref{thm:equivalence}. We may therefore transform the rates individually and aggregate them subsequently, with no loss of equivalence.

\subsection{Reformulating the SLqP Problem}\label{sec:reformulation}

Applying the QFT to the SINR of each rate in~\eqref{rate_function_cell_free}, with one auxiliary variable $x_{k}$ per UE collected in $\mathbf{X} = [x_{1},\ldots,x_{K}]^{T}$, yields the auxiliary problem
\begin{align}
\underset{\mathbf{X},\mathbf{P}}{\mathrm{maximize}}\quad
& f_{K_{q}}\left(\hat{r}_{1}\left(\mathbf{X},\mathbf{P}\right),\ldots,\hat{r}_{K}\left(\mathbf{X},\mathbf{P}\right)\right)
\tag{P2}\label{eq:reformulated_problem}
\\
\mathrm{subject\ to}\quad
& p_{k,b}\geq0,\quad k=1,\ldots,K,\; b=1,\ldots,B
\nonumber
\\
& \sum_{k\in\mathcal{K}_{b}}p_{k,b}\leq P_{\mathrm{max}},\quad b=1,\ldots,B
\nonumber
\end{align}
where the transformed rate is
\begin{align}\label{eq:rhat}
\hat{r}_{k}\left(\mathbf{X},\mathbf{P}\right)
&=
\mathrm{log}\bigg(1+2x_{k}\big|a_{k,k}(\mathbf{P})\big|
\nonumber \\
&
\qquad -
x_{k}^{2}\Big(\sum_{k'\in\mathcal{K}_{-k}}\big|a_{k,k'}(\mathbf{P})\big|^{2}+\sigma^{2}\Big)\bigg)
\end{align}
By Theorem~\ref{thm:equivalence}, Problem~\ref{eq:reformulated_problem} is equivalent to Problem~\ref{SLqP_original_problem}; the auxiliary variables enlarge the search space but alter neither the optimal value nor the optimal powers. With the powers held fixed, an optimal choice of the auxiliary variables follows from Lemma~\ref{lemma:QFT} as
\begin{align}\label{eq:X_opt}
x_{k}^{\mathrm{opt}}=\frac{\left|a_{k,k}(\mathbf{P})\right|}{\sum_{k'\in\mathcal{K}_{-k}}\left|a_{k,k'}(\mathbf{P})\right|^{2}+\sigma^{2}},\ \forall k
\end{align}
i.e., the ratio of the desired signal amplitude of UE $k$ to the total interference-plus-noise power it experiences, a per-UE weight reflecting how favourable that UE's current transmission conditions are. Each $x_{k}^{\mathrm{opt}}$ depends on the powers of all UEs, which is the mechanism by which the coupling of the network-wide objective enters the algorithm. As observed in \cite{PercI_10466635}, the SLqP utility is not strictly concave, so other choices of $\mathbf{X}$ may attain the same auxiliary objective value; \eqref{eq:X_opt} is nonetheless the choice that preserves equivalence with the original problem, and we adopt it throughout.

\subsubsection{Where the approach of \cite{PercI_10466635} breaks down}
In \cite{PercI_10466635}, the auxiliary objective is concave in the powers once the auxiliary variables are held fixed, and the cyclic maximization follows immediately. This is \emph{not} the case in~\eqref{eq:rhat}. The modulus $|a_{k,k}(\mathbf{P})|$ in the first term of~\eqref{eq:rhat} is, by the MRT structure in~\eqref{eq:MRT}, equal to $\sum_{b\in\mathcal{C}_{k}}\sqrt{p_{k,b}}\,\Vert\mathbf{h}_{k,b}\Vert_{2}$, a sum of square roots of the powers and hence concave; the difficulty lies in the interference term. Expanding $|a_{k,k'}(\mathbf{P})|^{2} = |\sum_{b'\in\mathcal{C}_{k'}}\sqrt{p_{k',b'}}\,\mathbf{h}_{k,b'}^{H}\mathbf{v}_{k',b'}|^{2}$ reveals cross products of the form $\sqrt{p_{k',b_{1}}p_{k',b_{2}}}$ between distinct APs $b_1 \ne b_2$ serving the same UE $k'$, each weighted by $\mathrm{Re}\{(\mathbf{h}_{k,b_{1}}^{H}\mathbf{v}_{k',b_{1}})(\mathbf{h}_{k,b_{2}}^{H}\mathbf{v}_{k',b_{2}})^{*}\}$, a coefficient of either sign. The cross products are geometric means and therefore concave in the powers, so the interference term is, in general, neither convex nor concave in $\mathbf{P}$, and the composition rule that underlies the approach in \cite{PercI_10466635} cannot be applied. These cross products are the mathematical signature of coherent joint transmission: they are absent whenever each UE is served by a single node, which is why the difficulty does not arise in the cellular setting.

\subsubsection{Restoring concavity in the power block}
To remedy this, we introduce the one-to-one mapping
\begin{align}\label{eq:theta_map}
\theta_{k,b}\doteq\sqrt{p_{k,b}},\quad k=1,\ldots,K,\; b=1,\ldots,B
\end{align}
which is a bijection on the non-negative orthant and therefore incurs no loss of equivalence. The variable $\theta_{k,b}$ is the transmit amplitude of AP $b$ toward UE $k$, i.e., the scale of the beamformer $\mathbf{v}_{k,b}$, so that~\eqref{eq:theta_map} is the parameterization in which the beamforming problems of \cite{PercII_10458868} are naturally posed. Collecting the amplitudes in $\mathbf{\Theta}$ and writing
\begin{align}\label{eq:abar}
\bar{a}_{k,k'}(\mathbf{\Theta}) \doteq \sum_{b \in\mathcal{C}_{k'}} \theta_{k',b}\, {\bf h}_{k,b}^H {\bf v}_{k',b}
\end{align}
for the precoded channel expressed in the amplitudes, so that $\bar{a}_{k,k'}(\mathbf{\Theta}) = a_{k,k'}(\mathbf{P})$ under~\eqref{eq:theta_map}, the transformed rate becomes
\begin{align}\label{eq:rcheck}
\check{r}_{k}\left(\mathbf{X},\mathbf{\Theta}\right)
&=
\mathrm{log}\bigg(1+2x_{k}\Big(\sum_{b\in\mathcal{C}_{k}}\theta_{k,b}\left\Vert \mathbf{h}_{k,b}\right\Vert _{2}\Big)
\nonumber \\
&
\qquad - x_{k}^{2}\Big(\sum_{k'\in\mathcal{K}_{-k}}\big|\bar{a}_{k,k'}(\mathbf{\Theta})\big|^{2}+\sigma^{2}\Big)\bigg)
\end{align}
where the first term follows from the MRT structure in~\eqref{eq:MRT}, since $|a_{k,k}(\mathbf{P})| = |\sum_{b\in\mathcal{C}_{k}}\sqrt{p_{k,b}}\,\mathbf{h}_{k,b}^{H}\mathbf{h}_{k,b}/\Vert \mathbf{h}_{k,b}\Vert _{2}| = \sum_{b\in\mathcal{C}_{k}}\theta_{k,b}\Vert \mathbf{h}_{k,b}\Vert _{2}$, every summand being non-negative. The modulus has been eliminated, what was a sum of square roots is now linear in $\mathbf{\Theta}$, and, crucially, the interference term is now the squared modulus of $\bar{a}_{k,k'}(\mathbf{\Theta})$, an \emph{affine} function of $\mathbf{\Theta}$.

The mapping in~\eqref{eq:theta_map} must also be propagated to the constraints, under which the per-AP budget becomes $\sum_{k\in\mathcal{K}_{b}}\theta_{k,b}^{2}\leq P_{\mathrm{max}}$ for $b=1,\ldots,B$. This is a convex quadratic constraint, equivalently the second-order cone constraint $\big\Vert \big[\theta_{k,b}\big]_{k \in \mathcal{K}_{b}}\big\Vert_{2}\le\sqrt{P_{\mathrm{max}}}$, and is retained exactly by a convex solver. We note that this constraint is not equivalent to the linear constraint $\sum_{k\in\mathcal{K}_{b}}\theta_{k,b}\le\sqrt{P_{\mathrm{max}}}$, which is strictly tighter whenever an AP serves more than one UE; imposing the latter would shrink the feasible set and forfeit equivalence with Problem~\ref{SLqP_original_problem}. We thus obtain the equivalent problem
\begin{align}
\underset{\mathbf{\Theta}}{\mathrm{maximize}}\quad
& f_{K_{q}}\left(\check{r}_{1}\left(\mathbf{X},\mathbf{\Theta}\right),\ldots,\check{r}_{K}\left(\mathbf{X},\mathbf{\Theta}\right)\right)
\tag{P3}\label{eq:theta_problem}
\\
\mathrm{subject\ to}\quad
& \theta_{k,b}\geq0,\quad k=1,\ldots,K,\; b=1,\ldots,B
\nonumber
\\
& \sum_{k\in\mathcal{K}_{b}}\theta_{k,b}^{2}\leq P_{\mathrm{max}},\quad b=1,\ldots,B
\nonumber
\end{align}

\begin{lemma}[Block-concavity]\label{lem:block_concavity}
The objective of Problem~\ref{eq:theta_problem} is concave in the transformed power variables $\mathbf{\Theta}$ when the auxiliary variables $\mathbf{X}$ are held fixed.
\end{lemma}

\begin{proof}
We apply the composition rule~\cite{boyd2004convex}: $h(g_{1}(\mathbf{x}),\ldots,g_{K}(\mathbf{x}))$ is concave in $\mathbf{x}$ provided that $h$ is concave and non-decreasing in each argument and each $g_{k}$ is concave. In~\eqref{eq:rcheck}, the term $2x_{k}\sum_{b}\theta_{k,b}\Vert\mathbf{h}_{k,b}\Vert_{2}$ is linear, hence concave, in $\mathbf{\Theta}$. Each interference term $|\bar{a}_{k,k'}(\mathbf{\Theta})|^{2}$ is the squared modulus of an affine function of $\mathbf{\Theta}$ and is therefore convex, so its negative scaled by $x_{k}^{2} \ge 0$ is concave. The argument of the logarithm is thus concave, and since the logarithm is concave and non-decreasing, $\check{r}_{k}(\mathbf{X},\mathbf{\Theta})$ is concave in $\mathbf{\Theta}$. Finally, by Lemma~\ref{lem:slqp}(i) and (iii), $f_{K_{q}}$ is concave and non-decreasing in each component. Identifying $f_{K_{q}}$ with $h$ and $\check{r}_{k}$ with $g_{k}$ completes the proof.
\end{proof}

\subsubsection{Epigraph form for a disciplined convex solver}
Problem~\ref{eq:theta_problem} is concave in $\mathbf{\Theta}$ for fixed $\mathbf{X}$, yet its objective remains non-smooth and is not expressed in a form that a disciplined convex programming parser can verify. We therefore introduce a slack variable $t_{k}$ per UE, collected in $\mathbf{t}=[t_{1},\ldots,t_{K}]^{T}$, and recast the rates as constraints:
\begin{align}
\underset{\mathbf{\Theta},\mathbf{t}}{\mathrm{maximize}}\quad
& f_{K_{q}}\left(\mathbf{t}\right)
\tag{P4}\label{eq:cvx_problem}
\\
\mathrm{subject\ to}\quad
& \theta_{k,b}\geq0,\quad k=1,\ldots,K,\; b=1,\ldots,B
\nonumber
\\
& \sum_{k\in\mathcal{K}_{b}}\theta_{k,b}^{2}\leq P_{\mathrm{max}},\quad b=1,\ldots,B
\nonumber
\\
& t_{k}\leq\check{r}_{k}\left(\mathbf{X},\mathbf{\Theta}\right),\quad k=1,\ldots,K
\nonumber
\end{align}
where the rate constraints are expressed in solver-ready form by exponentiating both sides, i.e.,
\begin{align}\label{eq:exp_constraint}
&1+2x_{k}\Big(\sum_{b\in\mathcal{C}_{k}}\theta_{k,b}\left\Vert \mathbf{h}_{k,b}\right\Vert _{2}\Big)
\nonumber \\
&\qquad
-x_{k}^{2}\Big(\sum_{k'\in\mathcal{K}_{-k}}\big|\bar{a}_{k,k'}(\mathbf{\Theta})\big|^{2}+\sigma^{2}\Big)
\geq\mathrm{exp}\left(t_{k}\right)
\end{align}
for $k=1,\ldots,K$. The left-hand side of~\eqref{eq:exp_constraint} is concave in $\mathbf{\Theta}$ and the right-hand side is convex in $t_{k}$, so the constraint is a valid convex restriction, and Problem~\ref{eq:cvx_problem} is a concave maximization that may be solved using CVX~\cite{cvx} or the subgradient method; its objective is available in the standard CVX library as the \texttt{sum\_smallest} atom. The relaxation of the rates into the constraints $t_{k}\leq\check{r}_{k}(\mathbf{X},\mathbf{\Theta})$ is tight at the optimum: by Lemma~\ref{lem:slqp}(iii) there is no incentive to leave any $t_{k}$ strictly below its upper bound, so every such constraint holds with equality at a maximizer and the optimal values of Problems~\ref{eq:cvx_problem} and~\ref{eq:theta_problem} coincide. Moreover, the combinatorial structure of the SLqP utility is handled implicitly: the solver at no point enumerates the $\binom{K}{K_{q}}$ selection vectors of~\eqref{eq:SLqP_2}, but operates on the concave atom directly, and the active set is revealed by the solution rather than searched over. This is precisely the combinatorial explosion that rendered the smoothed reformulation of Section~\ref{sec:problem_formulation} intractable.

\subsection{The Proposed Algorithm}\label{sec:algorithm}

\begin{algorithm}[t]
		\SetAlgoLined
		\SetInd{0.1em}{1em}
		\caption{QFT Algorithm for SLqP Rate Maximization in UC-mMIMO Networks}
        \label{algorithm:percentileAlgo}
		\KwIn{channel vectors $\{{\bf h}_{k,b}\}$, serving sets $\{\mathcal{C}_{k}\}$, percentile $q$, per-AP budget $P_{\max}$, tolerance $\epsilon$, iteration cap $I_{\max}$.}
		\KwOut{power allocation ${\bf P}^{\star}=\{p^{\star}_{k,b}\}$.}
		compute the percentile number $K_{q}$ (Section~\ref{sec:problem_formulation})\;
		compute the MRT directions ${\bf v}_{k,b}$ using~\eqref{eq:MRT} and the cross terms ${\bf h}_{k,b}^{H}{\bf v}_{k',b}$\;
		initialize $\theta_{k,b}\left[0\right]\leftarrow\sqrt{P_{\max}/|\mathcal{K}_{b}|}$ for all $b$ and all $k\in\mathcal{K}_{b}$\;
		evaluate $u\left[0\right]\leftarrow f_{K_{q}}\big({\bf r}({\bf \Theta}\left[0\right])\big)$ and set $i\leftarrow0$\;
		\Repeat{$\big|u\left[i\right]-u\left[i-1\right]\big|\le\epsilon\,\big|u\left[i-1\right]\big|$ \KwSty{or} $i=I_{\max}$}{
			$i\leftarrow i+1$\;
			update ${\bf X}\left[i\right]$ using~\eqref{eq:X_opt} for fixed ${\bf \Theta}\left[i-1\right]$\;
			update ${\bf \Theta}\left[i\right]$ by solving Problem~\ref{eq:cvx_problem} for fixed ${\bf X}\left[i\right]$\;
			evaluate $u\left[i\right]\leftarrow f_{K_{q}}\big({\bf r}({\bf \Theta}\left[i\right])\big)$\;
		}
		\Return $p^{\star}_{k,b}\leftarrow\big(\theta_{k,b}\left[i\right]\big)^{2}$\;
	\end{algorithm}

The preceding development yields a problem that is concave in $\mathbf{X}$ for fixed $\mathbf{\Theta}$ and, by Lemma~\ref{lem:block_concavity}, concave in $\mathbf{\Theta}$ for fixed $\mathbf{X}$, though not jointly concave in both. This block-concave structure admits the cyclic maximization summarized in Algorithm~\ref{algorithm:percentileAlgo}: the closed-form auxiliary update re-measures the position of each UE relative to the rest of the network and, by Theorem~\ref{thm:equivalence}, restores equality between the auxiliary and original objectives, and the power update then reallocates the transmit powers in light of that measurement. We initialize with equal power allocation, $p_{k,b}\left[0\right]=P_{\max}/|\mathcal{K}_{b}|$, which is feasible by construction, requires no channel knowledge and favours no UE a priori; since Problem~\ref{SLqP_original_problem} is non-convex, the limit point depends on the starting point, and this choice is reproducible. We terminate when the relative improvement in the SLqP utility falls below a tolerance $\epsilon$ or when an iteration cap $I_{\max}$ is reached, which is appropriate because the sequence $\{u[i]\}$ is non-decreasing and bounded above (Section~\ref{sec:convergence}); $u[i]$ is evaluated using the true rates $\mathbf{r}(\mathbf{\Theta}[i])$, so the reported value is always an achievable utility. Finally, the auxiliary variables are updated before the powers within each iteration, which guarantees that the auxiliary objective is tight at the start of every power update, as the convergence argument requires.

\begin{remark}[Implicit refinement of the serving clusters]\label{rem:refinement}
Since $f_{K_{q}}$ depends only on the $K_{q}$ smallest rates, a UE whose rate lies above the $K_{q}^{\rm th}$ order statistic has, by Lemma~\ref{lem:slqp}(iv), zero subgradient weight, and power assigned to it enters the optimization only through the interference it inflicts upon the UEs that do contribute. The power update therefore has an incentive to reduce $\theta_{k,b}$ on such links, and an AP $b$ may transmit negligible power to a UE $k$ even though $b \in \mathcal{C}_{k}$ under the geometric rule~\eqref{eq:C_k}: the algorithm refines the user-centric clusters in an interference-aware manner as a by-product, without any explicit cluster-selection variable. The refinement is self-limiting, since a UE whose rate falls among the lowest $K_{q}$ enters the active set and further reductions are penalized directly by the objective.
\end{remark}

\subsection{Convergence}\label{sec:convergence}

We now establish that Algorithm~\ref{algorithm:percentileAlgo} is non-decreasing in the SLqP rate objective and converges to a directional stationary point. The analysis must account for the change of variables in~\eqref{eq:theta_map}, which is absent from \cite{PercI_10466635}: the minorant that drives the algorithm is constructed in the transformed domain, and we must show that it is tangent to the non-smooth SLqP objective, including in its directional derivatives.

\begin{theorem}\label{thm:monotone}
Algorithm~\ref{algorithm:percentileAlgo} is non-decreasing in the objective of Problem~\ref{SLqP_original_problem} after each iteration, and the sequence of objective values converges.
\end{theorem}

\begin{proof}
The non-decreasing nature of the algorithm follows from the chain
\begin{subequations}\label{eq:convergence_chain}
\begin{align}
f_{K_{q}}\left(\mathbf{r}\left(\mathbf{P}\left[i\right]\right)\right)
&\overset{(a)}{=}
f_{K_{q}}\left(\hat{\mathbf{r}}\left(\mathbf{X}\left[i+1\right],\mathbf{P}\left[i\right]\right)\right)
\label{eq:conv_a}
\\
&\overset{(b)}{=}f_{K_{q}}\left(\check{\mathbf{r}}\left(\mathbf{X}\left[i+1\right],\mathbf{\Theta}\left[i\right]\right)\right)
\label{eq:conv_b}
\\
&\overset{(c)}{\leq} f_{K_{q}}\left(\check{\mathbf{r}}\left(\mathbf{X}\left[i+1\right],\mathbf{\Theta}\left[i+1\right]\right)\right)
\label{eq:conv_c}
\\
&\overset{(d)}{=}f_{K_{q}}\left(\hat{\mathbf{r}}\left(\mathbf{X}\left[i+1\right],\mathbf{P}\left[i+1\right]\right)\right)
\label{eq:conv_d}
\\
&\overset{(e)}{\leq}
f_{K_{q}}\left(\mathbf{r}\left(\mathbf{P}\left[i+1\right]\right)\right)
\label{eq:conv_e}
\end{align}
\end{subequations}
where $\mathbf{P}[i]$, $\mathbf{\Theta}[i]$ and $\mathbf{X}[i]$ denote the values of these variables in the $i^{\rm th}$ iteration. Step~(a) follows from Theorem~\ref{thm:equivalence}, since $\mathbf{X}[i+1]$ is obtained from $\mathbf{P}[i]$ through the maximizer in~\eqref{eq:X_opt}, at which the auxiliary objective is tight; (b) and (d) follow since the mapping in~\eqref{eq:theta_map} is a bijection between the non-negative powers and $\mathbf{\Theta}$, so the two parameterizations describe the same point; (c) follows since $\mathbf{\Theta}[i+1]$ maximizes the concave Problem~\ref{eq:cvx_problem} for fixed $\mathbf{X}[i+1]$, for which $\mathbf{\Theta}[i]$ is feasible and the epigraph constraints are tight; and (e) follows from the bound in~\eqref{eq:aggregator_bound}, which holds for every $\mathbf{X}$. Convergence follows since the per-AP budget in Problem~\ref{SLqP_original_problem} limits the total radiated power, so every rate, and hence the SLqP objective, is bounded above, and a non-decreasing sequence that is bounded above converges.
\end{proof}

Theorem~\ref{thm:monotone} does not characterize the limit. Since $f_{K_{q}}$ is not differentiable by Lemma~\ref{lem:slqp}(iv), the usual first-order condition does not apply, and the relevant concept is \emph{directional} stationarity, i.e., a point at which the directional derivative of the objective is non-positive along every feasible direction; as noted in \cite{PercI_10466635}, this is a stronger requirement than the stationarity guarantees of the prior FP and WMMSE literature, which address smooth weighted sum-rate objectives: let $\mathcal{T}$ denote the feasible set of Problem~\ref{eq:theta_problem}, let $g(\mathbf{\Theta}) := f_{K_{q}}(\mathbf{r}(\mathbf{\Theta}))$ denote the SLqP rate objective in the amplitude variables, and define the surrogate $u(\mathbf{\Theta}\,|\,\mathbf{\Theta}') := f_{K_{q}}(\check{\mathbf{r}}(\mathbf{X}^{\mathrm{opt}}(\mathbf{\Theta}'),\mathbf{\Theta}))$, where $\mathbf{X}^{\mathrm{opt}}(\mathbf{\Theta}')$ is the update in~\eqref{eq:X_opt} evaluated at $\mathbf{\Theta}'$. The power update of Algorithm~\ref{algorithm:percentileAlgo} is exactly $\mathbf{\Theta}[i+1] \in \argmax_{\mathbf{\Theta} \in \mathcal{T}} u(\mathbf{\Theta}\,|\,\mathbf{\Theta}[i])$.

\begin{theorem}\label{thm:dstationary}
Algorithm~\ref{algorithm:percentileAlgo} is a cyclic minorization-maximization (MM) algorithm, and every limit point of the sequence $\{\mathbf{\Theta}[i]\}$ it generates is a directional stationary point of the SLqP rate maximization problem in the amplitude variables, Problem~\ref{eq:theta_problem}.
\end{theorem}

\begin{proof}
{By~\cite[Theorem~1]{razaviyayn2013unified} (see also~\cite{sun2017majorization})}, every limit point of an MM sequence is a directional stationary point provided that the feasible set is closed and convex and the surrogate satisfies, for all $\mathbf{\Theta},\mathbf{\Theta}' \in \mathcal{T}$: (A1)~$u(\mathbf{\Theta}'\,|\,\mathbf{\Theta}') = g(\mathbf{\Theta}')$; (A2)~$u(\mathbf{\Theta}\,|\,\mathbf{\Theta}') \le g(\mathbf{\Theta})$; (A3)~$u'(\mathbf{\Theta}';\mathbf{d}\,|\,\mathbf{\Theta}') = g'(\mathbf{\Theta}';\mathbf{d})$ for every feasible direction $\mathbf{d}$ at $\mathbf{\Theta}'$, where $u'$ and $g'$ denote directional derivatives; and (A4)~$u$ is continuous in $(\mathbf{\Theta},\mathbf{\Theta}')$. The set $\mathcal{T}$ defined by the constraints of Problem~\ref{eq:theta_problem} is non-empty, closed, convex and bounded, so the sequence $\{\mathbf{\Theta}[i]\}$ has at least one limit point. Conditions (A1) and (A2) are the content of Theorem~\ref{thm:equivalence} transported through the bijection~\eqref{eq:theta_map}: by~\eqref{eq:QFT_lower_bound}, $\check{r}_{k}(\mathbf{X}^{\mathrm{opt}}(\mathbf{\Theta}'),\mathbf{\Theta}) \le r_{k}(\mathbf{\Theta})$ for every $\mathbf{\Theta}$, with equality at $\mathbf{\Theta} = \mathbf{\Theta}'$ by the tightness of the bound at $x^{\mathrm{opt}}$, and $f_{K_{q}}$ is non-decreasing. Condition (A4) holds because $\mathbf{X}^{\mathrm{opt}}(\cdot)$ in~\eqref{eq:X_opt} is continuous, $\check{r}_{k}$ in~\eqref{eq:rcheck} is continuous in $(\mathbf{X},\mathbf{\Theta})$ wherever the argument of its logarithm is positive, which it is in a neighbourhood of every point $(\mathbf{\Theta}',\mathbf{\Theta}')$ since there it equals $1+\mathrm{SINR}_{k}(\mathbf{\Theta}') \ge 1$, and $f_{K_{q}}$ is continuous by Lemma~\ref{lem:slqp}(ii).

Condition (A3) is where the non-smoothness enters, and we verify it directly. Fix $\mathbf{\Theta}' \in \mathcal{T}$ and write $\boldsymbol{\rho}(\mathbf{\Theta}) := \mathbf{r}(\mathbf{\Theta})$ and $\check{\boldsymbol{\rho}}(\mathbf{\Theta}) := \check{\mathbf{r}}(\mathbf{X}^{\mathrm{opt}}(\mathbf{\Theta}'),\mathbf{\Theta})$, so that $g = f_{K_{q}} \circ \boldsymbol{\rho}$ and $u(\cdot\,|\,\mathbf{\Theta}') = f_{K_{q}} \circ \check{\boldsymbol{\rho}}$. Both inner maps are continuously differentiable in a neighbourhood of $\mathbf{\Theta}'$: the argument of the logarithm in~\eqref{rate_function_cell_free} is a ratio of polynomials in $\mathbf{\Theta}$ whose denominator is bounded below by $\sigma^{2} > 0$, and the argument of the logarithm in~\eqref{eq:rcheck} is a polynomial in $\mathbf{\Theta}$ that is positive near $\mathbf{\Theta}'$ as noted above. Moreover, $\rho_{k} - \check{\rho}_{k} \ge 0$ on this neighbourhood with equality at $\mathbf{\Theta}'$, so the differentiable function $\rho_{k} - \check{\rho}_{k}$ attains its minimum at $\mathbf{\Theta}'$ and its gradient vanishes there; hence the two inner maps share the same Jacobian $\mathbf{J} \in \mathbb{R}^{K \times N}$ at $\mathbf{\Theta}'$, and the same value $\mathbf{z} := \boldsymbol{\rho}(\mathbf{\Theta}') = \check{\boldsymbol{\rho}}(\mathbf{\Theta}')$. 

Next, let $\mathcal{A}^{\star}(\mathbf{z}) \subseteq \mathcal{A}_{K_{q}}$ denote the set of selection vectors attaining the minimum in~\eqref{eq:SLqP_2} at $\mathbf{z}$. Since $\mathcal{A}_{K_{q}}$ is finite, every ${\bf a} \notin \mathcal{A}^{\star}(\mathbf{z})$ satisfies ${\bf a}^{T}\mathbf{z} \ge f_{K_{q}}(\mathbf{z}) + \delta$ for some $\delta > 0$, and therefore $f_{K_{q}}(\mathbf{z} + \mathbf{w}) = f_{K_{q}}(\mathbf{z}) + \min_{{\bf a} \in \mathcal{A}^{\star}(\mathbf{z})} {\bf a}^{T}\mathbf{w}$ for all $\mathbf{w}$ of sufficiently small norm. Let $\mathbf{d}$ be a feasible direction at $\mathbf{\Theta}'$ and $\alpha \downarrow 0$. Using $\boldsymbol{\rho}(\mathbf{\Theta}'+\alpha\mathbf{d}) = \mathbf{z} + \alpha\mathbf{J}\mathbf{d} + o(\alpha)$, this expansion and the Lipschitz continuity of $f_{K_{q}}$, we obtain $g'(\mathbf{\Theta}';\mathbf{d}) = \lim_{\alpha \downarrow 0} [g(\mathbf{\Theta}'+\alpha\mathbf{d}) - g(\mathbf{\Theta}')]/\alpha = \min_{{\bf a} \in \mathcal{A}^{\star}(\mathbf{z})} {\bf a}^{T}\mathbf{J}\mathbf{d}$. The identical computation for $u(\cdot\,|\,\mathbf{\Theta}')$ involves the same $\mathbf{z}$, the same $\mathbf{J}$ and hence the same active set $\mathcal{A}^{\star}(\mathbf{z})$, and yields the same value. This establishes (A3) and completes the proof.
\end{proof}

\begin{remark}\label{rem:amplitude_param}
The amplitude parameterization in~\eqref{eq:theta_map} is the one in which the rates are smooth functions of the design variables, since the powers enter~\eqref{rate_function_cell_free} only through their square roots, and it is the parameterization in which the beamforming problems of \cite{PercII_10458868} are posed; Theorem~\ref{thm:dstationary} is accordingly stated in terms of $\mathbf{\Theta}$. On the open positive orthant the mapping~\eqref{eq:theta_map} is a diffeomorphism, so a limit point with strictly positive powers is also a directional stationary point of Problem~\ref{SLqP_original_problem} in the original power variables; at a point where some powers vanish, the Jacobian of the mapping is singular and the two notions need not coincide. Since Problem~\ref{SLqP_original_problem} is NP-hard, a directional stationary point is moreover the best that can be achieved in polynomial time, and Theorem~\ref{thm:dstationary} guarantees that the algorithm cannot stall at an arbitrary point.
\end{remark}

\subsection{Percentile Optimization for the Benchmark Architectures}\label{sec:benchmark_architectures}

In Section~\ref{sec:results} we compare UC-mMIMO against two benchmark architectures, massive MIMO and small cells. So that the comparison isolates the effect of the architecture, the powers in all three cases are optimized for the same SLqP utility rather than for a proxy. Here we show that doing so requires no new machinery: in both benchmarks the algorithm of Section~\ref{sec:algorithm} reduces to the QFT algorithm of \cite{PercI_10466635}, which also makes precise which part of our development is specific to the cell-free setting.

The essential distinction is that both benchmarks are single-node architectures: in massive MIMO each UE is served by the macro BS of the cell in which it resides, and in the small-cell architecture by the single AP whose Voronoi region contains it, so that $|\mathcal{C}_{k}| = 1$ for every UE. Denoting by $c_{k}$ the node serving UE $k$, the cluster sums in~\eqref{rate_function_cell_free} collapse to single terms, and the rate of UE $k$ becomes $r_{k}^{\rm (sn)}(\mathbf{P})=\log(1+A_{k}(\mathbf{P})/B_{k}(\mathbf{P}))$, subject to the per-node budget $\sum_{k: c_{k} = c} p_{k,c} \le \bar{P}_{\max}$ at every node $c$, where $\bar{P}_{\max}$ is specified for each benchmark in Section~\ref{sec:results}, and where
\begin{align}
A_{k}\left(\mathbf{P}\right) &= p_{k,c_{k}}\big|\mathbf{v}_{k,c_{k}}^{H}\mathbf{h}_{k,c_{k}}\big|^{2}
\label{eq:A_single}
\\
B_{k}\left(\mathbf{P}\right) &= \sum_{k'\in\mathcal{K}_{-k}} p_{k',c_{k'}}\big|\mathbf{v}_{k',c_{k'}}^{H}\mathbf{h}_{k,c_{k'}}\big|^{2}+\sigma^{2}
\label{eq:B_single}
\end{align}
Observe that $A_{k}(\mathbf{P})$ is linear in the powers and $B_{k}(\mathbf{P})$ is affine. Consequently $2x_{k}\sqrt{A_{k}(\mathbf{P})}$ is concave, being the square root of a linear function, while $-x_{k}^{2}B_{k}(\mathbf{P})$ is affine, so the argument of the logarithm in the auxiliary rate is concave in $\mathbf{P}$ directly, and by the composition rule together with Lemma~\ref{lem:slqp}(i) and (iii) the auxiliary objective is concave in the powers when the auxiliary variables are fixed; the per-node budget likewise remains linear. In other words, the block-concavity that Section~\ref{sec:reformulation} had to recover through the change of variables in~\eqref{eq:theta_map} is already present, and that step may simply be omitted. What remains is exactly the cyclic maximization of \cite{PercI_10466635}: the auxiliary variables are updated in closed form as $x_{k}=\sqrt{A_{k}(\mathbf{P})}/B_{k}(\mathbf{P})$, and the powers are obtained by solving a concave program for fixed $\mathbf{X}$. Algorithm~\ref{algorithm:percentileAlgo} thus applies verbatim with $\mathbf{\Theta}$ replaced by $\mathbf{P}$ and the final squaring step removed, and Theorems~\ref{thm:monotone} and~\ref{thm:dstationary} hold, with steps~\eqref{eq:conv_b} and~\eqref{eq:conv_d} of the convergence chain becoming vacuous.

\subsection{Computational Complexity and Fronthaul Load}\label{sec:complexity}

Two costs determine whether Algorithm~\ref{algorithm:percentileAlgo} is practical: the arithmetic performed at the CU and the information that must traverse the fronthaul. Let $\bar{C}=\max_{k}|\mathcal{C}_{k}|$ denote the largest cluster size and $N=\sum_{k\in\mathcal{K}}|\mathcal{C}_{k}|=\sum_{b\in\mathcal{B}}|\mathcal{K}_{b}|\le K\bar{C}$ the total number of power variables, which is the natural measure of problem size in a user-centric deployment.

\begin{table*}[t]
	\centering
	\caption{Simulation parameters.}
	\footnotesize
	\begin{tabular}{|p{0.22\linewidth}|p{0.23\linewidth}||p{0.29\linewidth}|p{0.14\linewidth}|}
		\hline
		\multicolumn{1}{|l|}{ \textit{\textbf{Parameter}}} & \multicolumn{1}{l||}{\textit{\textbf{Value}}}&\multicolumn{1}{l|}{ \textit{\textbf{Parameter}}} & \multicolumn{1}{l|}{\textit{\textbf{Value}}}\\
		\hline
		Number of APs: $B$ & \{MDD, UDD\}$=\{75, 150\}$ APs
		&
        Number of virtual cells
		  &
        $5\times 5 = 25$ regions
		\\
        Number of UEs: $K$ & $100$ UEs &
        Side length of square virtual cell $r_{\rm vcell}$
        &
        $200$~m
        \\
		Number of antennas per AP: $M$ & $8$ antennas
		&
        Path loss threshold $\rho$;$\ \ \ $ distance threshold $\bar{d}$
        &
        ${\rm PL}(\bar{d})$;$\ \ \ $ $r_{\rm vcell}/2$
        \\
		Serving set $\mathcal{C}_{k}$ of UE $k$
		&
        Formula~\eqref{eq:C_k}
        &
		Path loss terms: $\alpha_{\rm pl}$, $d_0$
		& $3.8$, $1.1~{\rm m}$
		\\
        Percentile $q$ & 10
        &
        Noise spectral density
		& $-150~{\rm dBm/ Hz}$
        \\
        Power budget of AP $P_{\rm max}$ & $36~{\rm dBm}$
		&
		Std. deviation of shadowing: $\sigma_\mathrm{sh|dB}$& $4~{\rm dB}$
		\\
		Bandwidth & $20~{\rm MHz}$
        &
		Noise figure
		& $9~{\rm dB}$
        \\
		Carrier frequency: $f_{\rm c}$ & $1.8~{\rm GHz}$
		&
        AP-UE height difference $d_{\rm h}$ &
        13.5 m
		\\
		\hline
	\end{tabular}
	\label{table:sim_parameters}
\end{table*}

\subsubsection{Computational complexity}
The cross terms $\mathbf{h}_{k,b}^{H}\mathbf{v}_{k',b}$ depend only on the channels and the MRT directions and not on the powers, so they are computed once per coherence block, at a cost of $\mathcal{O}(KNM)$ inner products of length $M$, and reused in every iteration. Given these terms, the auxiliary update~\eqref{eq:X_opt} costs $\mathcal{O}(KN)$ per iteration, with no matrix inversion or eigendecomposition at any stage, a direct consequence of fixing the beamforming directions to MRT and optimizing only the powers. The power update dominates: Problem~\ref{eq:cvx_problem} is a concave program in $n=N+K$ variables, subject to $N$ non-negativity constraints, $B$ second-order cone constraints from the per-AP budgets and $K$ exponential-cone constraints from~\eqref{eq:exp_constraint}, which a generic interior-point method solves in $\mathcal{O}(n^{3.5})$ arithmetic operations, up to a logarithmic factor in the solver accuracy~\cite{boyd2004convex}. The overall cost of Algorithm~\ref{algorithm:percentileAlgo} is therefore $\mathcal{O}(KNM)+\mathcal{O}(I\,(N+K)^{3.5})$, where $I$ denotes the number of outer iterations. Since the cluster sizes in a user-centric deployment are bounded by the association rule in~\eqref{eq:C_k} and do not grow with the network, $N$ scales linearly rather than quadratically in the number of UEs, and this cost is polynomial in the network size notwithstanding the NP-hardness of the underlying problem. We note that this is a generic interior-point bound: Problem~\ref{eq:cvx_problem} is sparse, since the constraint of each UE involves only its own cluster and the clusters of its interferers, and a solver exploiting this sparsity may be expected to perform considerably better. As shown in Section~\ref{sec:results}, $I$ is small: the algorithm attains more than 95\% of its final SLqP rate within five iterations in all the configurations we consider.

\subsubsection{Fronthaul load}
The CU implements the optimization and computes the powers used by the APs, which keeps the AP design simple and inexpensive and corresponds to a higher ``Option Number'' for the functional split of the protocol stack, such as Option~7.2 or Option~8~\cite[Chapter~1]{ammar2023analysis}. The associated signalling can be counted precisely. Since the precoders are given by MRT and are therefore computable locally, AP $b$ can form the scalar $\mathbf{h}_{k,b}^{H}\mathbf{v}_{k',b}=\mathbf{h}_{k,b}^{H}\mathbf{h}_{k',b}/\Vert\mathbf{h}_{k',b}\Vert_{2}$ from its own measurements, for each UE $k'\in\mathcal{K}_{b}$ that it serves and each UE $k$ it can hear, and forward these $K|\mathcal{K}_{b}|$ complex scalars in place of the $M$-dimensional vectors from which they were computed; the channel gains $\Vert\mathbf{h}_{k,b}\Vert_{2}$ required by~\eqref{eq:rcheck} are recovered as the $k'=k$ entries of the same set. The uplink fronthaul load in complex scalars per coherence block is therefore $\sum_{b\in\mathcal{B}}K|\mathcal{K}_{b}|=KN$, as compared with $\mathcal{O}(KBM)$ had the raw channels been forwarded, a reduction by a factor on the order of $M$ for typical cluster sizes; in the downlink direction, the CU returns to each AP only the powers of the UEs it serves, for a total of $N$ real scalars. Both quantities are independent of the iteration count $I$: every iteration executes entirely at the CU using data already in its possession, so the fronthaul is exercised once per coherence block rather than once per iteration, which decouples the convergence speed of the algorithm from the fronthaul capacity of the deployment, typically the scarcer resource in a dense cell-free network. A decentralized implementation would be attractive for very large deployments, but is fundamentally more difficult here than for sum-rate or max-min objectives, since an AP cannot establish from local information alone whether the UEs it serves belong to the active set; decentralized percentile optimization remains an open problem.

\begin{figure*}[t]
	\centering
	\begin{subfigure}[t]{0.48\textwidth}
	\centering
    \includegraphics[width=1\textwidth]{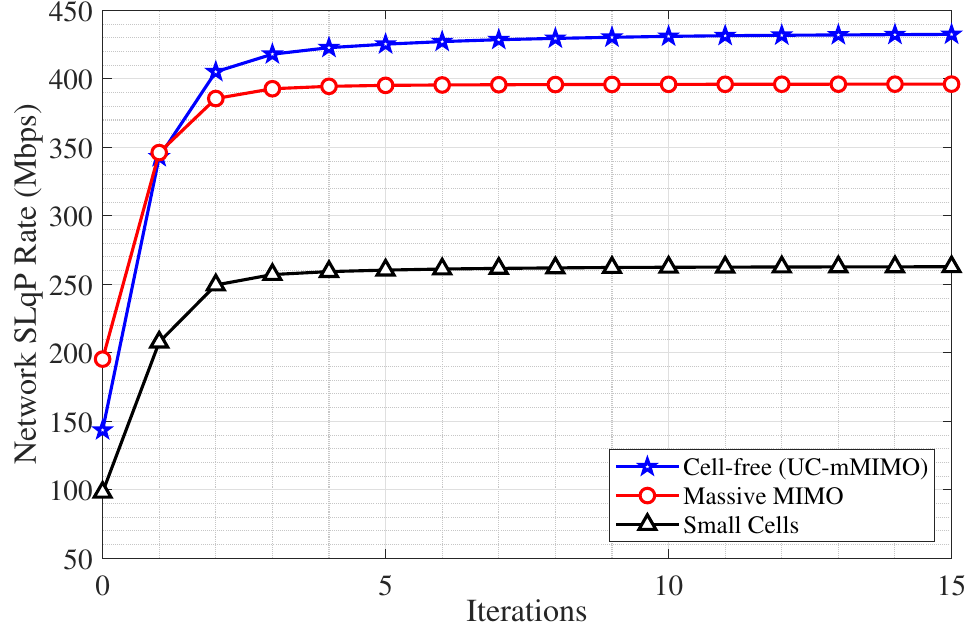}
	\caption{MDD, $B = 75$ APs.}
	\label{sfig:conv_APs75}
\end{subfigure}
		\hspace{1.5em}%
	\begin{subfigure}[t]{0.48\textwidth}
	\centering
    \includegraphics[width=1\textwidth]{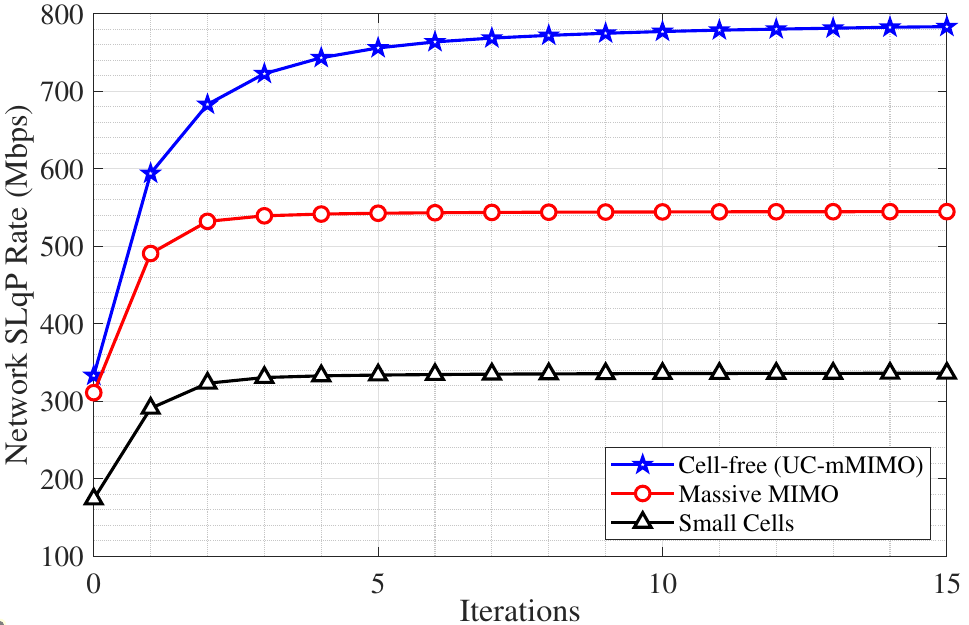}
	\caption{UDD, $B = 150$ APs.}
	\label{sfig:conv_APs150}
    \end{subfigure}
	\caption{Convergence of the network SLqP rate under Algorithm~\ref{algorithm:percentileAlgo} for the three architectures, starting from equal power allocation.}
	\label{fig:conv}
\end{figure*}

\begin{figure*}[t]
	\centering
	\begin{subfigure}[t]{0.48\textwidth}
	\centering
    \includegraphics[width=1\textwidth]{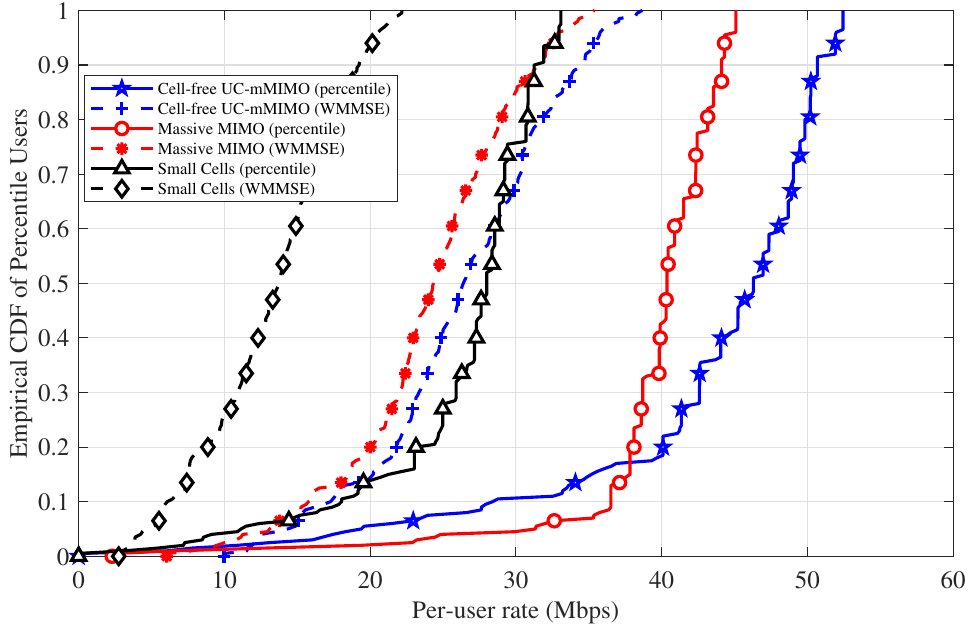}
	\caption{MDD, $B = 75$ APs.}
	\label{sfig:CDF_75}
\end{subfigure}
		\hspace{1.5em}%
	\begin{subfigure}[t]{0.48\textwidth}
	\centering
    \includegraphics[width=1\textwidth]{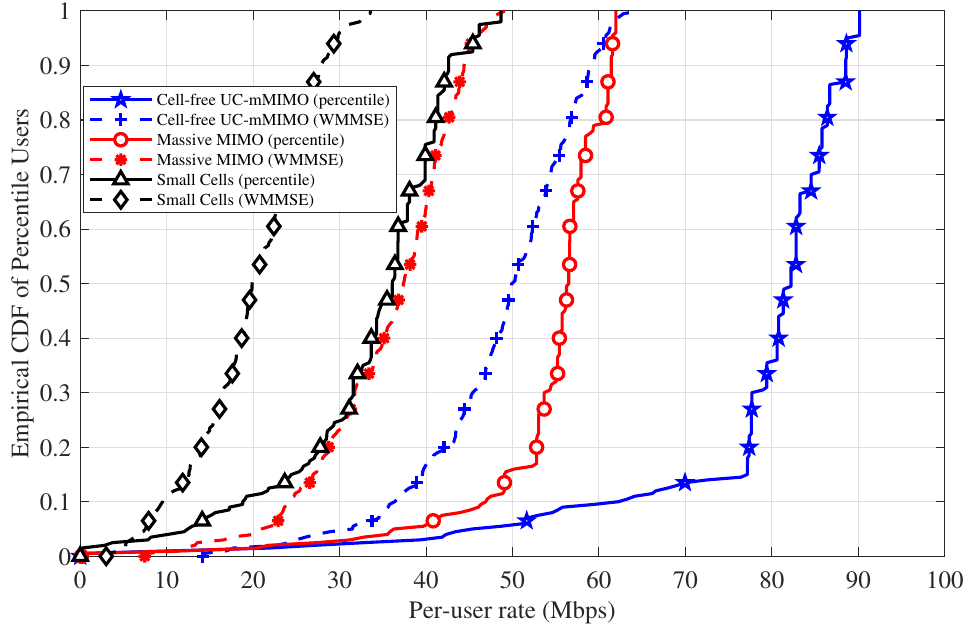}
	\caption{UDD, $B = 150$ APs.}
	\label{sfig:CDF_150}
    \end{subfigure}
	\caption{Empirical CDF of the rates of the lowest-$10^{\rm th}$-percentile UEs for the three architectures under percentile optimization and WMMSE.}
	\label{fig:CDF_percentile}
\end{figure*}

\section{Numerical Results}\label{sec:results}

\subsection{Simulation Setup}

We consider a network area of $1\times1~\text{km}^2$ containing $K = 100$~UEs, divided into $5\times5$ square virtual cells of side $r_{\rm vcell} = 200$~m, i.e., $C = 25$ regions. In each virtual cell we place $4$ uniformly distributed UEs and $B_{c}$ uniformly distributed APs, each equipped with $M = 8$ antennas. We consider two levels of densification: a moderately dense deployment (MDD) with $B_{c} = 3$, i.e., $B = 75$~APs, and an ultra-dense deployment (UDD) with $B_{c} = 6$, i.e., $B = 150$~APs. Wraparound is applied to the locations of the APs and UEs so as to eliminate border effects and emulate a network of infinite extent. 

We set $q = 10$, so that the SLqP utility targets the $K_{q} = 10$ UEs that would be identified as cell-edge UEs in a conventional network, and form the serving clusters according to~\eqref{eq:C_k} with the distance threshold $\bar{d} = r_{\rm vcell}/2$. The LSF coefficient between UE $k$ and AP $b$ comprises path loss and shadowing, $\beta_{k,b} = {\rm PL}(\bar{d}_{k,b})\, \psi_{k,b}$, where the path loss follows the COST231 Walfisch-Ikegami model~\cite{Walfisch14401},
${\rm PL}(\bar{d}_{k,b}) =
\big(\sqrt{(\bar{d}_{k,b})^2 + d_{\rm h}^2}/d_0\big)^{-\alpha_{\rm pl} }$, with $\bar{d}_{k,b}$ the horizontal distance between AP $b$ and UE $k$, $d_{\rm h}$ the fixed height difference between the two, and $\alpha_{\rm pl}$ and $d_0$ the path loss exponent and reference distance, and the shadowing is log-normal, $\psi_{k,b} = 10^{x_\psi/10}$ with $x_\psi\sim \mathcal{N}(0,\sigma_\mathrm{sh|dB}^2)$. Table~\ref{table:sim_parameters} summarizes the parameters; any departure from them is stated in the corresponding figure.

The two benchmark architectures of Section~\ref{sec:benchmark_architectures} are generated from the same AP and UE locations. For massive MIMO, the virtual cells serve as the cells of a cellular layout, and the antennas of the $B_{c}$ APs in each cell are co-located at a macro BS at the cell center, which is therefore equipped with $\bar{M} = B_{c} M \in \{24, 48\}$ antennas in the MDD and the UDD, respectively; massive MIMO is thus the co-located counterpart of UC-mMIMO~\cite{6810508}. To keep the comparison fair, each macro BS is assigned the sum of the power budgets of the APs it replaces, $\bar{P}_{\max} = B_{c} P_{\max}$, which is reasonable since APs consume far less transmit power than macro BSs. For small cells, each AP serves the UEs in its Voronoi region with its own budget $\bar{P}_{\max} = P_{\max}$, and no cooperation takes place.

\begin{figure*}[t]
	\centering
	\begin{subfigure}[t]{0.48\textwidth}
	\centering
    \includegraphics[width=1\textwidth]{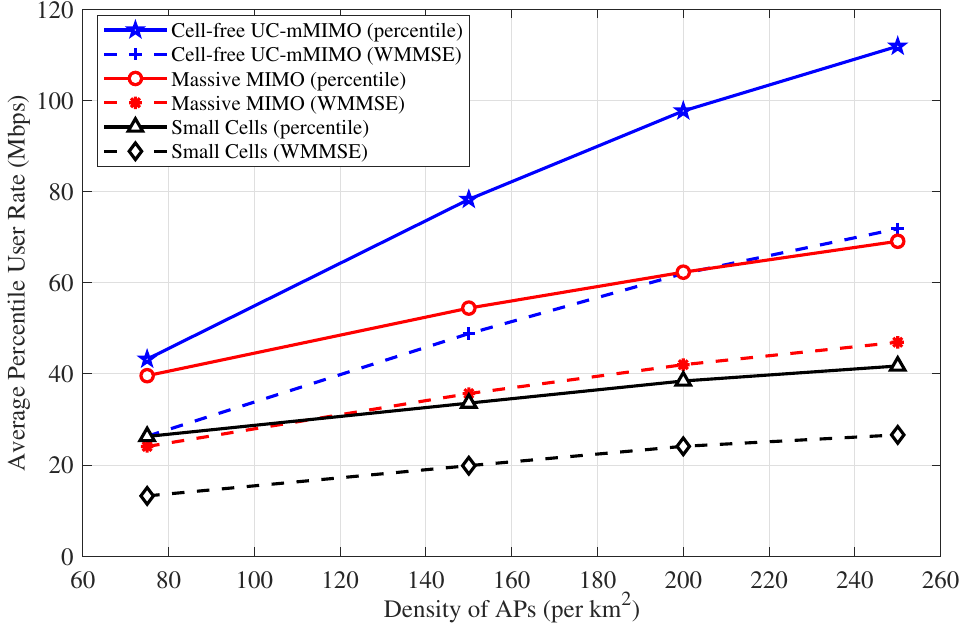}
	\caption{versus the AP density.}
	\label{sfig:avg_vs_density}
\end{subfigure}
		\hspace{1.5em}%
	\begin{subfigure}[t]{0.48\textwidth}
	\centering
    \includegraphics[width=1\textwidth]{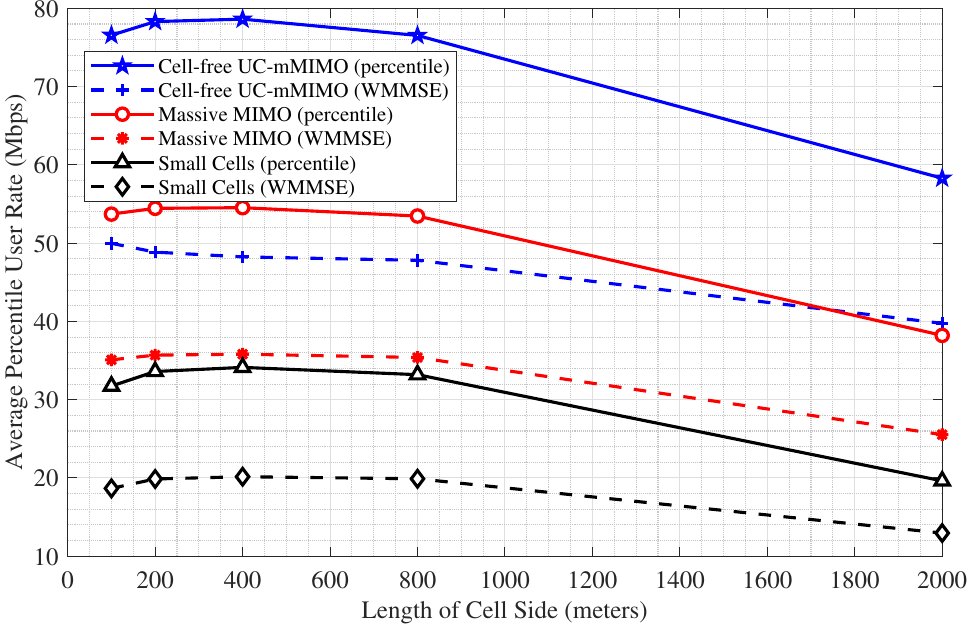}
	\caption{versus the side length of the virtual cell, with the number of APs and UEs per virtual cell held fixed (UDD).}
	\label{sfig:avg_vs_cellsize}
    \end{subfigure}
	\caption{Average rate of the lowest-$10^{\rm th}$-percentile UEs.}
	\label{fig:avg_percentile}
\end{figure*}

Two comparisons are of interest, and we organize the results accordingly. The first is between resource allocation schemes: within each architecture, we compare the proposed SLqP power control against the WMMSE algorithm~\cite{5756489}, which is the standard benchmark for weighted sum-rate maximization and performs comparably to the FP approach of~\cite{8314727}. The second is between architectures: UC-mMIMO, massive MIMO and small cells are compared under the same SLqP utility, which isolates the effect of the architecture from that of the resource allocation. All results are averaged over independent network realizations.

\subsection{Convergence}

Fig.~\ref{fig:conv} plots the network SLqP rate, i.e., the sum of the $K_{q} = 10$ lowest UE rates, against the iteration index for the MDD and the UDD. In every case the objective increases monotonically from its equal-power starting value and settles, in agreement with Theorem~\ref{thm:monotone}, and this behaviour is observed for all three architectures despite the different SINR regimes, AP distributions and degrees of cooperation that they entail. Convergence is also fast: more than 95\% of the final SLqP rate is attained within five iterations in all cases, and the two single-node benchmarks have essentially converged after three to five iterations. UC-mMIMO takes slightly longer, roughly ten iterations to come within 1\% of its final value in the UDD, which is to be expected since the coherent combination across clusters couples many more power variables.

Fig.~\ref{fig:conv} also carries a message that anticipates the results below. At iteration zero, i.e., under equal power allocation, UC-mMIMO is \emph{not} the best architecture for the lowest-percentile UEs: in the MDD it starts below massive MIMO (145 versus 195~Mbps), and in the UDD the two are comparable. It is the optimization that unlocks the cell-free advantage, with the SLqP rate of UC-mMIMO growing by a factor of about three in the MDD and 2.4 in the UDD, against factors of about two and 1.75 for massive MIMO. The architecture alone does not deliver lower-percentile performance; resource allocation is what converts the architecture's potential into rates for these percentile UEs.

\begin{figure*}[t]
	\centering
	\begin{subfigure}[t]{\textwidth}
	\centering
    \includegraphics[width=1\textwidth]{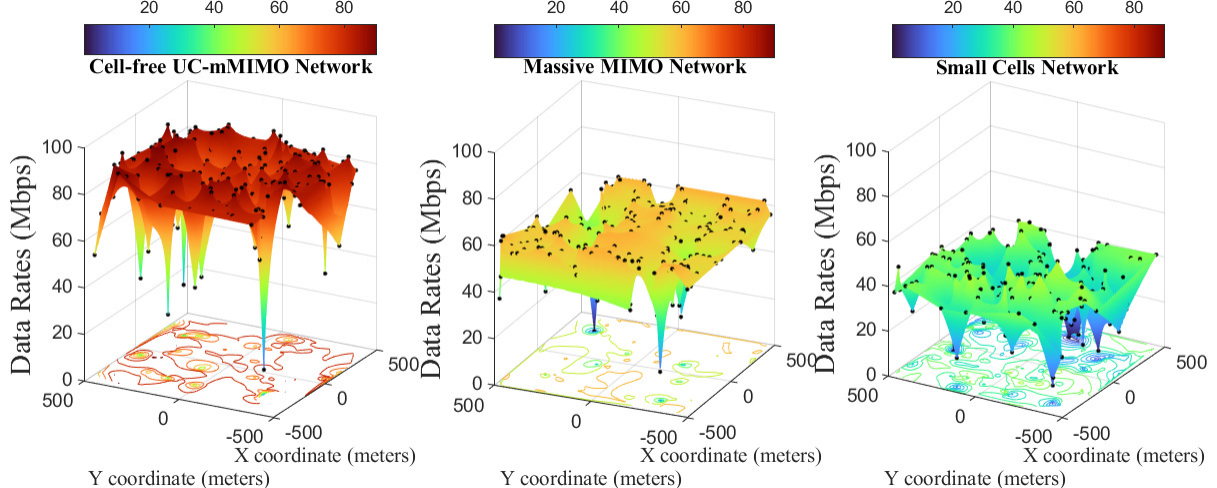}
	\caption{Rate surface interpolated between the percentile UEs; the interpolation serves visualization only.}
	\label{sfig:spatial_surface}
\end{subfigure}
		\\[0pt] 
	\begin{subfigure}[t]{\textwidth}
	\centering
    \includegraphics[width=1\textwidth]{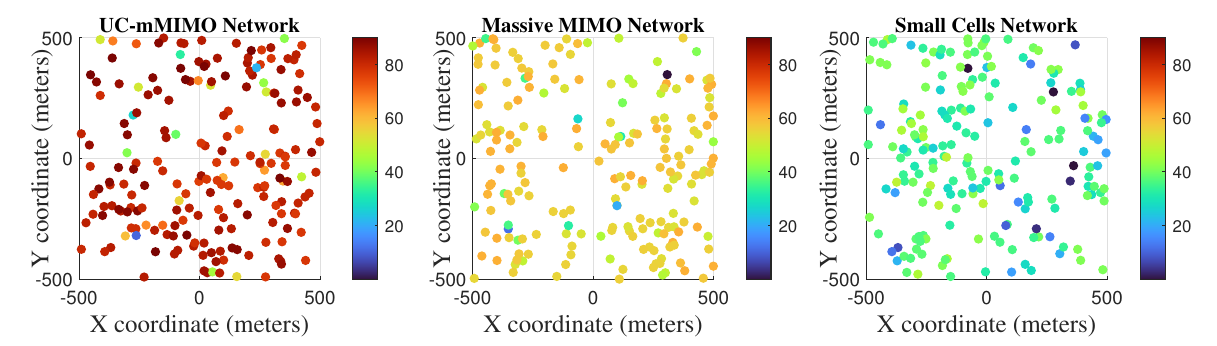}
	\caption{Locations of the percentile UEs, coloured by their achieved rate in Mbps; no interpolation.}
	\label{sfig:spatial_points}
    \end{subfigure}
	\caption{Spatial distribution of the rates of the lowest-$10^{\rm th}$-percentile UEs under percentile optimization in the UDD, pooled over network realizations, for the three architectures. The colour scale is common to all panels.}
	\label{fig:spatialDataRates}
\end{figure*}

\subsection{Value of Percentile Optimization in Cell-Free Networks}

We now address the question raised in the introduction: does explicitly targeting the lowest-percentile UEs still matter in an architecture without cell edges? Fig.~\ref{fig:CDF_percentile} plots the empirical CDF of the rates of the lowest-$10^{\rm th}$-percentile UEs, pooled over network realizations, for the MDD and the UDD, and Fig.~\ref{fig:avg_percentile}(\subref{sfig:avg_vs_density}) plots the average rate of these UEs as the AP density is increased from 75 to 250~APs/km$^2$. In each figure, solid curves correspond to the proposed percentile optimization and dashed curves to WMMSE; the answer to the question is given by the gap between the solid and dashed curves of the same colour.

The gap is large. For UC-mMIMO, the median rate of the percentile UEs rises from about 25~Mbps under WMMSE to about 45~Mbps under percentile optimization in the MDD, and from about 50 to about 81~Mbps in the UDD. In terms of the average rate of the percentile UEs in Fig.~\ref{fig:avg_percentile}(\subref{sfig:avg_vs_density}), percentile optimization improves on WMMSE by 72\% at 75~APs/km$^2$, 59\% at 150, 58\% at 200 and 56\% at 250~APs/km$^2$. The improvement is of the same order in the benchmark architectures, 47--67\% for massive MIMO and 56--100\% for small cells, so the value of percentile optimization is not an artifact of any one architecture.

The shape of the CDFs reveals the mechanism. Under WMMSE, the rates of the low-percentile UEs are spread over a wide range, since a sum-rate-oriented utility has no reason to equalize them, whereas under percentile optimization the CDF is nearly vertical over most of its range, most visibly for UC-mMIMO in the UDD, where more than 80\% of the percentile UEs are confined to the interval from roughly 77 to 90~Mbps: the SLqP objective transfers power from UEs above the $K_{q}^{\rm th}$ order statistic to those below it until the lowest rates are equalized, which is precisely the behaviour described in Remark~\ref{rem:refinement}. The lowest tail of each CDF corresponds to the few UEs with very poor channels toward every AP in their cluster, and it is on these UEs that the max-min utility would have spent the entire network budget.

Fig.~\ref{fig:avg_percentile}(\subref{sfig:avg_vs_density}) exhibits a further comparison that we regard as the clearest evidence of the value of percentile optimization. The dashed blue curve of UC-mMIMO under WMMSE lies well below the solid red curve of massive MIMO under percentile optimization at 75 and 150~APs/km$^2$, and essentially coincides with it at 200 and 250~APs/km$^2$. That is, a conventional co-located massive MIMO network with percentile-optimized power control delivers about as much to its lowest-percentile UEs as a cell-free network with WMMSE, and at moderate densities considerably more: the choice of resource allocation is worth as much as the choice of architecture, and only their combination, the solid blue curve, realizes the full gain.

\subsection{Comparison of the Network Architectures}

We next compare the three architectures under the same SLqP utility, i.e., the solid curves of Figs.~\ref{fig:conv}--\ref{fig:avg_percentile}. Since the three use the same AP and UE locations, the same total number of antennas and the same total transmit power, the differences are attributable to the architecture alone.

UC-mMIMO is the superior architecture at every operating point considered. At convergence in Fig.~\ref{fig:conv}, its network SLqP rate exceeds that of massive MIMO by 9\% in the MDD (432 versus 395~Mbps) and by 44\% in the UDD (782 versus 543~Mbps), and that of small cells by 65\% and 131\%, respectively. The median rate of the low-percentile UEs in Fig.~\ref{fig:CDF_percentile} tells an identical story: 45, 40 and 28~Mbps for UC-mMIMO, massive MIMO and small cells respectively in the MDD, and 81, 55 and 36~Mbps in the UDD. We note that in the MDD, massive MIMO is marginally ahead of UC-mMIMO in the lowest 15\% of the low-percentile UEs, i.e., for the few UEs in each realization that are poorly served by every AP within their cluster distance, and that this crossover disappears in the UDD, where the UC-mMIMO curve lies to the right of the massive MIMO curve throughout.

The decisive observation is how the architectures respond to densification, shown in Fig.~\ref{fig:avg_percentile}(\subref{sfig:avg_vs_density}). As the density increases from 75 to 250~APs/km$^2$, the average rate of the percentile UEs grows by a factor of 2.6 for UC-mMIMO, 1.7 for massive MIMO and 1.6 for small cells. The advantage of UC-mMIMO over massive MIMO consequently widens from 8\% at 75~APs/km$^2$ to 44\%, 58\% and 62\% at 150, 200 and 250~APs/km$^2$, and over small cells from 1.7$\times$ to 2.7$\times$. Small cells gain the least, since densification brings interferers as close as serving APs and the single-AP association cannot exploit the additional infrastructure; massive MIMO gains more, through the larger co-located array; UC-mMIMO gains the most, since every additional AP enlarges the clusters of the UEs around it and contributes coherently to their signals rather than only to interference. UC-mMIMO is therefore the architecture that converts densification into lower-percentile throughput most effectively, which is the regime toward which future networks are moving~\cite{10522673}.

Fig.~\ref{fig:avg_percentile}(\subref{sfig:avg_vs_cellsize}) varies the side length of the virtual cell from 100~m to 2000~m with the number of APs and UEs per virtual cell held fixed, so that all inter-node distances scale with the cell side. For cell sides up to 800~m the rates are essentially flat, since the network is interference limited and the SINRs are invariant to a uniform scaling of the geometry; beyond this point the network becomes noise limited and the rates fall. The ordering of the architectures is preserved throughout, with UC-mMIMO ahead of massive MIMO by 42--54\% and of small cells by 2.3--3$\times$, and the advantage is largest in the noise-limited regime at 2000~m, where proximity to several serving APs matters most.

The CDFs and averages above summarize the rates of the percentile UEs but say nothing about \emph{where} these UEs are, which is the question an operator planning coverage must answer. Fig.~\ref{fig:spatialDataRates} therefore plots, for the UDD under percentile optimization, the locations of the lowest-$10^{\rm th}$-percentile UEs pooled over the network realizations, coloured by their achieved rate on a scale common to the three architectures; Fig.~\ref{fig:spatialDataRates}(\subref{sfig:spatial_surface}) interpolates a rate surface between these UEs to make the spatial structure visible, and Fig.~\ref{fig:spatialDataRates}(\subref{sfig:spatial_points}) shows the individual UEs without interpolation. Three observations follow. First, in every architecture the percentile UEs are spread across the entire area rather than confined to particular regions. This is expected for UC-mMIMO, which has no cell boundaries by construction, but it also holds for the two cellular benchmarks once their powers are optimized for the same percentile objective: the UEs that remain in the lowest percentile after optimization are those with poor channels to every node within reach, whose positions are dictated by shadowing and local geometry rather than by the cell layout. Second, the level and the uniformity of the colours reproduce, in space, the ordering established above. For UC-mMIMO, the surface is a plateau at the top of the scale, with almost every percentile UE at 80--90~Mbps and only isolated dips; for massive MIMO it is a lower plateau at 50--65~Mbps; for small cells it is lower still, at 25--45~Mbps, and markedly rougher, with many UEs below 25~Mbps and several close to zero. The dips in the surfaces are the lowest tail of the CDFs in Fig.~\ref{fig:CDF_percentile}, and they are isolated points rather than contiguous regions in all three cases: percentile optimization leaves no coverage holes of spatial extent, only individual UEs that cannot be rescued by any power allocation.

Third, the figure can be read directly as an outage map. A service with a rate requirement $R_{\rm th}$ is in outage at every UE whose colour lies below $R_{\rm th}$ on the scale, and the comparison across architectures is stark. For a requirement of 40~Mbps, only a handful of the percentile UEs in UC-mMIMO and massive MIMO would be in outage, against roughly half of those in small cells; for a requirement of 70~Mbps, roughly one in ten of the percentile UEs in UC-mMIMO would be in outage, against all of them in massive MIMO and small cells. Since the percentile UEs are, by definition, the weakest tenth of the network, these fractions bound the outage of the network as a whole: with percentile optimization, UC-mMIMO can support a 70~Mbps service with an outage of about 1\% of the UEs in the UDD, a level that neither benchmark architecture can approach, which is the operational meaning of the gains reported in Fig.~\ref{fig:avg_percentile}.

\begin{figure}
	\centering
	\begin{subfigure}[t]{0.49\linewidth}
		\centering
		\includegraphics[width=1\textwidth]{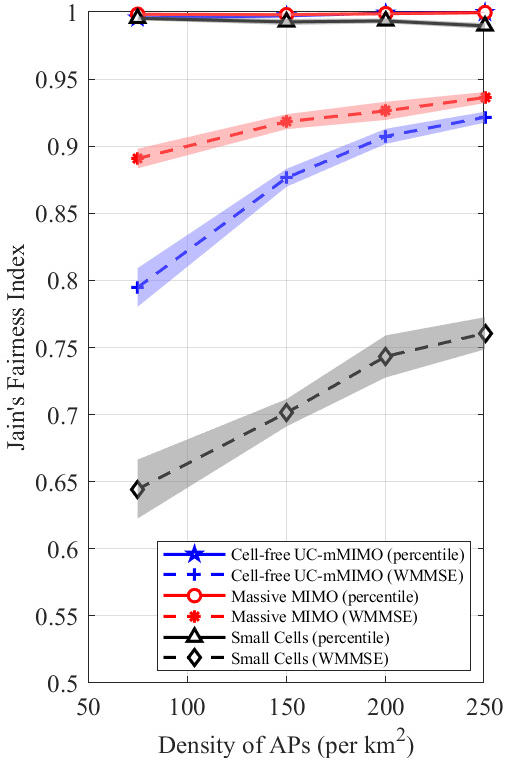}
		\caption{versus AP density}
		\label{sfig:fairness_vs_APDensity_withVariance_mini}
	\end{subfigure}
	\begin{subfigure}[t]{0.49\linewidth}
		\centering
		\includegraphics[width=1\textwidth]{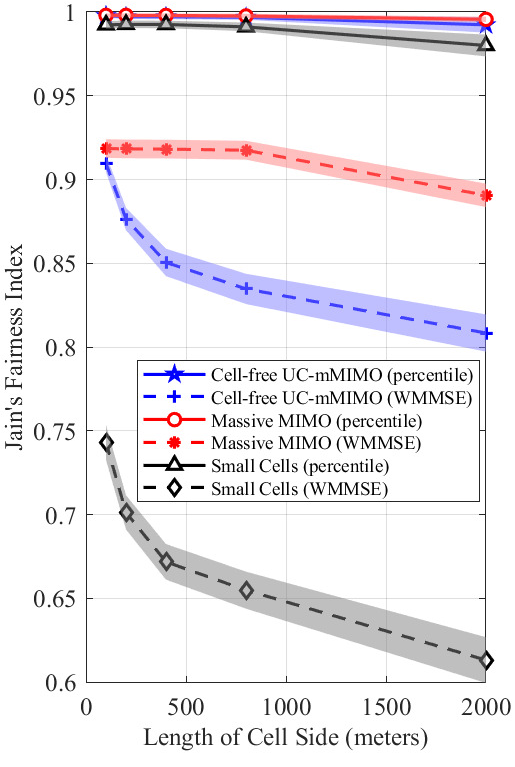}
		\caption{versus virtual cell side}
		\label{sfig:fairness_vs_CellSize_withVariance_mini}
	\end{subfigure}
	\caption{Jain's fairness index for the three architectures under percentile optimization and WMMSE. Markers show network realizations averages; shaded regions show the range.}
    \label{fig:fairness_withVariance}
\end{figure}

\subsection{User Fairness}

To measure fairness among the UEs we use Jain's fairness index~\cite{jain1999throughput}, $J(\mathcal{K}) = ( \sum_{k = 1}^{K} r_k )^2 / (K \sum_{k = 1}^{K} r_k^2)$, which is bounded between $1/K$, attained when a single UE receives all the throughput, and $1$, attained when all UEs receive the same rate, and admits the interpretation that $J(\mathcal{K}) = x\%$ means that $x\%$ of the UEs are treated fairly.

Fig.~\ref{fig:fairness_withVariance} plots Jain's index against the AP density and the virtual cell side, showing both the average over network realizations and the range. Percentile optimization yields an index between 99\% and 99.8\% for all three architectures at every AP density, and between 97.3\% and 99.7\% at every cell size. WMMSE is far less fair, with an index ranging from 62.2\% to 89.8\% across the architectures at 75~APs/km$^2$ and from 72.7\% to 93.3\% at 200~APs/km$^2$, and from 73.2\% to 92.4\% at a cell side of 100~m and from 59.9\% to 89.7\% at 2000~m. 

Two features of these results deserve comment. First, under WMMSE the fairness of UC-mMIMO is \emph{below} that of massive MIMO at every density and cell size, and the fairness of every architecture degrades as the cell grows; the cell-free architecture does not, by itself, deliver fairness. Second, under percentile optimization the three architectures are indistinguishable in fairness, which, read together with Fig.~\ref{fig:avg_percentile}, means that UC-mMIMO attains the same near-perfect fairness as its competitors at substantially higher rates, using the same network resources. For a network operator, the combination therefore provides the fairest allocation of resources across the whole coverage area at the highest rates at which that fairness can be delivered, reducing the likelihood of outage for any UE while other UEs are allocated rates beyond their requirements.

\section{Conclusion}\label{sec:conclusion}

In this paper, we formulated and solved the problem of maximizing the sum-least-$q^{\rm th}$-percentile rate via downlink power control in user-centric cell-free massive MIMO networks, thereby combining a resource allocation strategy that explicitly targets the lowest-percentile UEs with an architecture designed around fairness across UEs. 

We showed that the cyclic maximization approach of \cite{PercI_10466635} does not apply directly, because coherent joint transmission introduces cross products of the powers that destroy block-concavity, and that a change of variables to the transmit amplitudes restores it. The resulting algorithm updates the auxiliary variables in closed form and the powers through a concave program, is non-decreasing in the SLqP rate, and converges to a directional stationary point, as we proved in a self-contained manner; in single-node architectures it reduces to the algorithm of \cite{PercI_10466635}.

Our simulation results support two conclusions. First, percentile optimization has real value even in a network with no cell edges: it improves the average rate of the lowest-$10^{\rm th}$-percentile UEs by 56--72\% over the WMMSE benchmark and raises Jain's fairness index above 0.99, and a massive MIMO network under percentile optimization delivers about as much to these UEs as a cell-free network under WMMSE. Second, among the three architectures compared under the same utility, a cell-free approach is the superior design: it delivers the highest lower-percentile rates at every AP density and cell size, its advantage over massive MIMO grows from 8\% to 62\% as the density increases from 75 to 250~APs/km$^2$, and it does so at the same fairness as its competitors. 

Decentralized percentile optimization, in which the global nature of the percentile objective must be reconciled with local processing at the APs, and the joint optimization of beamforming and power in the cell-free setting remain open problems for future work.

\footnotesize
\bibliography{P_Ref}
\bibliographystyle{ieeetr}

\vfill

\end{document}